\documentclass[12pt]{amsart}
\usepackage[a4paper, total={16cm,23cm}]{geometry}

\usepackage{amsmath,amssymb,amsfonts}
\usepackage{graphicx}
\usepackage{float}

\usepackage{pst-all}
\usepackage{pst-solides3d}

\usepackage[titletoc]{appendix}
\usepackage{xcolor}
\usepackage{longtable}
\usepackage{textcomp}

\usepackage[author-year,initials,nobysame]{amsrefs}
\usepackage{pst-all}
\usepackage{pstricks-add}

\def\const{\mathrm{const}}

\def\Real{\mathrm{Re\,}}
\def\Pic{\mathrm{Pic}}

\def\PV{\mathrm{P}_{\mathrm{V}}}

\numberwithin{equation}{section}

\newtheorem{theorem}{Theorem}[section]
\newtheorem{proposition}[theorem]{Proposition}
\newtheorem{corollary}[theorem]{Corollary}
\newtheorem{lemma}[theorem]{Lemma}
\newtheorem{remark}[theorem]{Remark}

\newtheorem{definition}[theorem]{Definition}

\title[Asymptotics of $\PV$]{Asymptotic behaviour of the fifth Painlev\'e transcendents in the space of initial values}
\author{Nalini Joshi}
\address{School of Mathematics and Statistics F07, The University of Sydney, New South Wales 2006, Australia}
\email{nalini.joshi@sydney.edu.au}

\author{Milena Radnovi\'c}
\address{School of Mathematics and Statistics F07, The University of Sydney, New South Wales 2006, Australia}
\email{milena.radnovic@sydney.edu.au}

\thanks{This research was supported by an Australian Laureate Fellowship \# FL120100094 from the Australian Research Council. The research of M.R.~was partially supported by by the Serbian Ministry of Education and Science (Project no. 174020: Geometry and Topology of Manifolds and Integrable Dynamical Systems).
}

\date{}

\begin{document}

\maketitle

\begin{abstract}
We study the asymptotic behaviour of the solutions of the fifth Painlev\'e equation as the independent variable approaches zero and infinity in the space of initial values.
We show that the limit set of each solution is compact and connected and, moreover, that any solution with an essential singularity at zero has an infinite number of poles and zeroes, and any solution with an essential singularity at infinity has infinite number of poles and, moreover, takes the value unity infinitely many times.
\end{abstract}


\section{Introduction}
Following the construction of initial value spaces of Painlev\'e equations by \citeauthor{Okamoto1979} \ycites{Okamoto1979}, and more recent asymptotic analysis of the solutions of the first, second and fourth Painlev\'e equations in such spaces, we investigate the solutions of the fifth Painlev\'e equation:
\begin{equation}\label{eq:PV}
\begin{split}
\PV\ : \ \frac{d^2y}{dx^2}=&\left(\frac{1}{2y}+\frac{1}{y-1}\right) \left(\frac{dy}{dx}\right)^2-\frac{1}{x}\frac{dy}{dx}+\frac{(y-1)^2}{x^2}\left(\alpha y+\frac{\beta}{y}\right)\\
&\qquad+\frac{\gamma y}{x}+\frac{\delta y(y+1)}{y-1},
\end{split}
\end{equation}
in an asymptotic limit in its initial value space. Complete information about the limit sets of transcendental solutions and their behaviours near the infinity set are found. Unlike earlier asymptotic investigations of $\PV$, we do not impose any reality constraints; here $y$ is a function of the complex variable $x$, and $\alpha$, $\beta$, $\gamma$, $\delta$ are given complex constants.

Noting that $\PV$ has (fixed) essential singularities only when the independent variable $x$ takes the values $0$ and $\infty$, we investigate the behaviour of the solutions near $x=0$. A similar analysis can be carried out near $\infty$ and we also include an outline of the main results for this limit.  We show that each solution that is singular at $x=0$ has infinitely many poles and zeroes in every neighbourhood of this point. Similarly, each solution singular at $x=\infty$ has infinitely many poles and, moreover, takes the value $1$ infinitely many times in each neighbourhood of infinity.

The starting point for our analysis is the compactification and regularisation of the initial-value space. To make explicit analytic estimates possible, we calculate detailed information about the Painlev\'e vector field after each resolution (or blow-up) of this space. A similar approach was carried out for the first, second, and fourth Painlev\'e equations respectively by \fullocite{DJ2011}, \fullocite{HJ2014}, and \fullocite{JR2016}. However, the construction of the initial-value spaces in each of these earlier works consisted of exactly nine blow-ups, while in the present paper, we will see that eleven blow-ups are needed, followed by two blow-downs.
The initial-value space is then obtained by removing the set, denoted by $\mathcal{I}$, of points which are not attained by any solution.

The main results obtained in this paper fall into four parts:
\vspace{2pt}
\begin{list}{}
  {\usecounter{enumi}
    \setlength{\parsep}{2pt}
    \setlength{\leftmargin}{12pt}\setlength{\rightmargin}{12pt}
    \setlength{\itemindent=-12pt}
  }

\item {\em Existence of a repeller set:} Theorem \ref{th:estimates} in Section \ref{sec:infinity} shows that $\mathcal I$ is a repeller for the flow. The theorem also provides the range of the independent variable for which a solution may remain in the vicinity of $\mathcal{I}$.
\item {\em Numbers of poles and zeroes:} In Corollary \ref{cor:infinity}, we prove that each solution that is sufficiently close to $\mathcal{I}$ has a pole in a neighbourhood of the corresponding value of the independent variable. Moreover, Theorem \ref{th:zeroespoles} shows that each solution with essential singularity at $x=0$ has infinitely many poles and infinitely many zeroes in each neighbourhood of that point.
\item {\em The complex limit set:} We prove in Theorem \ref{th:limit} that the limit set for each solution is non-empty, compact, connected, and invariant under the flow of the autonomous equation obtained as $x\to0$.
\item {\em Asymptotic behaviour as $x\to\infty$:} We show in Section \ref{sec:x_infinity} that each solution with an essential singularity at $x=\infty$ has infinitely many poles and takes the value unity infinitely many times in each neighbourhood of that point.
\end{list}
\vspace{2pt}

The asymptotic analysis of the fifth Painlev\'e transcendent has been studied by many authors, including
\citeauthor{AK2000} \ycites{AK1997a,AK1997b,AK2000}, \ocites{ZZ2016,BP2012,QS2006}, \citeauthor{LM1999} \ycites{LM1999,LM1999b}, \citeauthor{McCoyTang1986} \ycites{McCoyTang1986,McCoyTang1986b,McCoyTang1986c}, and \ocite{Jimbo1982}.
However, the literature on the asymptotic behaviours of the fifth Painlev\'e transcendent concentrates on behaviours on the real line, often focusing on special behaviours or solutions, while we consider all solution behaviours for $x\in\mathbb C$.
For other mathematical results related to $\PV$, see \cites{BFSVZ2013,Shimomura2011,KO2007,Sasaki2007,Clark2005b,LS2004,GJP2001,GJP2001b},
while for applications in physics see \cites{JMMS1980,Dyson1995,  Schief1994}, and references therein.

This paper is organised as follows.
In Section \ref{sec:space}, we construct and describe the space of the initial values for equation (\ref{eq:PV}), with
complete details of all the necessary calculations provided in Appendix \ref{sec:resolution}.
In Section \ref{sec:special}, we consider the special solutions of $\PV$.
Section \ref{sec:infinity} contains the analysis of the behaviours of solutions near the infinity set in the space of initial values. Results on the complex limit sets of solutions when the independent variable approaches $0$ are provided in Section \ref{sec:limit}.
The behaviours of the fifth Painlev\'e transcendent in the limit $x\to\infty$ are outlined in Section \ref{sec:x_infinity}.
A summary of the notation used in this paper is given in Appendix \ref{sec:notation}.

\section{The space of initial values}
\label{sec:space}
Since the fifth Painlev\'e equation is a second-order ordinary differential equation, solutions are locally defined by two initial values. Therefore, the space of initial values is two complex-dimensional. However, standard existence and uniqueness theorems only cover values of $y$ that are not arbitrarily close to $0$, $1$ or infinity (where the second derivative given by $\PV$ becomes ill-defined). In this section, we explain how to construct a regularized, compactified space of all possible initial values that overcomes these issues.

We start by formulating $\PV$ as an equivalent system of equations in Section \ref{sec:system} and describing its autonomous limiting form obtained as $x\to0$ in Section \ref{sec:auto}. The mathematical construction of the space of initial values is then given in Section \ref{sec:okamoto}. Where $y$ is arbitrarily close to the singular values $0$, $1$, $\infty$, the solutions have singular power series expansions, which become regularized Taylor expansions in corresponding domains of the initial value space. These regular expansions are provided in Section \ref{sec:poles}.

\subsection{A system equivalent to $\PV$}\label{sec:system}

With the change of the independent variable $t=\log x$, Equation (\ref{eq:PV}) becomes:
\begin{equation}\label{eq:PVlog}
\begin{split}
\frac{d^2y}{dt^2}=
\left(\frac{1}{2y}+\frac{1}{y-1}\right) \left(\frac{dy}{dt}\right)^2+(y-1)^2\left(\alpha y+\frac{\beta}{y}\right)
+e^t\gamma y +e^{2t}\frac{\delta y(y+1)}{y-1}.
\end{split}
\end{equation}
We rewrite Equation (\ref{eq:PVlog}) in the following way:
\begin{equation}\label{eq:PVlog-system}
\begin{split}
\frac{dy}{dt}=&
 2 y (y-1)^2z-(\theta_0+\eta)y^2+(2\theta_0+\eta-\theta_1 e^t)y-\theta_0,
\\
\frac{dz}{dt}=&-
(y-1)(3y-1)z^2+\big(2(\theta_0+\eta)y-2\theta_0-\eta+\theta_1 e^t\big)z
-\frac12\epsilon(\theta_0+\eta-\theta_{\infty}),
\end{split}
\end{equation}
where
$\theta_{\infty}^2=2\alpha$,
$\theta_0^2=-2\beta$,
$\theta_1^2=-2\delta$ $(\theta_1\neq0)$,
$\eta=-\frac{\gamma}{\theta_1}-1$, and
$\epsilon=\frac12(\theta_0+\theta_{\infty}+\eta)$.

\begin{remark}
Here, we assumed that $\delta\neq0$, which is a generic case of the fifth Painlev\'e equation.
When $\delta=0$, the $\PV$ is equivalent to the third Painlev\'e equation \cite{OO2006}.
\end{remark}

The system (\ref{eq:PVlog-system}) is Hamiltonian:
$$
\frac{dy}{dt}=\frac{\partial H}{\partial z},
\qquad
\frac{dz}{dt}=-\frac{\partial H}{\partial y},
$$
with Hamiltonian function:
\begin{equation}\label{eq:H}
H=y(y-1)^2z^2-(\theta_0+\eta)y^2z+(2\theta_0+\eta-\theta_1 e^t)yz-\theta_0z+\frac12\epsilon(\theta_0+\eta-\theta_{\infty})y.
\end{equation}

\subsection{The autonomous equation}
\label{sec:auto}

The autonomous equation corresponding to (\ref{eq:PVlog}) is:
\begin{equation}\label{eq:PVlog-auto}
\frac{d^2y}{dt^2}=\left(\frac{1}{2y}+\frac{1}{y-1}\right) \left(\frac{dy}{dt}\right)^2+(y-1)^2\left(\alpha y+\frac{\beta}{y}\right),
\end{equation}
which is equivalent to the autonomous version of (\ref{eq:PVlog-system}):
\begin{equation}\label{eq:PVlog-system-auto}
\begin{aligned}
&\frac{dy}{dt}=
(y-1)^2(2yz-\theta_0)+\eta y(y-1),
\\
&\frac{dz}{dt}=
(y-1)z(2\eta+2\theta_0+z-3yz)+\eta z-\frac12\epsilon(\theta_0+\eta-\theta_{\infty}).
\end{aligned}
\end{equation}
System (\ref{eq:PVlog-system-auto}) is also Hamiltonian:
$$
\frac{dy}{dt}=\frac{\partial E}{\partial z},
\qquad
\frac{dz}{dt}=-\frac{\partial E}{\partial y},
$$
with Hamiltonian:
\begin{equation}\label{eq:E}
E=y(y-1)^2z^2-(\theta_0+\eta)y^2z+(2\theta_0+\eta)yz-\theta_0z+\frac12\epsilon(\theta_0+\eta-\theta_{\infty})y.
\end{equation}

Using the first equation of (\ref{eq:PVlog-system-auto}) to express $z$, and using the fact that $E$ is constant along solutions, we get:
$$
\left(\frac{dy}{dt}\right)^2=(y-1)^2(4 C y + \theta_0^2 - 2 \theta_0 (\eta + \theta_0) y  + \theta_{\infty}^2  y^2 ),
\qquad
C=\const.
$$

It is worth observing that the constant function $y\equiv1$ is the only solution of this equation taking the value $1$.
From (\ref{eq:PVlog-system-auto}), the corresponding function $z$ is the solution of
$$
\frac{dz}{dt}=
\eta z-\frac12\epsilon(\theta_0+\eta-\theta_{\infty}).
$$
That is, we have
$$
z=c_1 e^{\eta t}+\frac{\epsilon(\theta_0+\eta-\theta_{\infty})}{2\eta},
$$
where $c_1$ is a constant.

The flow (\ref{eq:PVlog-system-auto}) has four fixed points: 
\begin{equation*}
(y, z)=
  \begin{cases}
   \left(1, \dfrac{\epsilon(\theta_0+\eta-\theta_{\infty})}{2\eta}\right)\\
   \\
   \left(Y_i, \dfrac{\theta_0}{2Y_i}+\dfrac{\eta}{2(1-Y_i)}\right), i\in\{1,2,3\},
  \end{cases}
\end{equation*}
where $Y_1$, $Y_2$, $Y_3$ are the roots of the following cubic polynomial in $Y$:
\[
  \left(-\theta_{\infty}^2 -   6 \eta^2 + 8 \eta \theta_0 + 2\theta_0^2\right) Y^3
 +\left( 2 \eta^2 - 12 \eta \theta_0 - 2\theta_0^2+\theta_{\infty}^2 \right) Y^2 
+\theta_0 (4 \eta - \theta_0) Y
+\theta_0^2.
\]


\subsection{Resolution of singularities}\label{sec:okamoto}
In this section, we explain how to construct the space of initial values for the system (\ref{eq:PVlog-system}).
First, we motivate the reason for this construction before introducing the notion of initial value spaces in Definition \ref{def:initial-values-space}, which is based on foliation theory.
We then explain how to construct such a space by carrying out resolutions or blow-ups, which are described in Definition \ref{def:blow-up}.

The system (\ref{eq:PVlog-system}) is a system of two first-order ordinary differential equations for $(y(t), z(t))$. 
Given initial values $(y_0, z_0)$ at $t_0$, local existence and uniqueness theorems provide a solution that is defined on a local polydisk $U\times V$ in $\mathbf C\times \mathbf C^2$, where $t_0\in U\subset \mathbf C$ and $(y_0, z_0)\in V\subset \mathbf C^2$.
Our interest lies in global extensions of these local solutions.

However, the occurence of movable poles in the Painlev\'e transcendents acts as a barrier to the extension of $U\times V$ to the whole of $\mathbf C\times \mathbf C^2$.
The first step to overcome this obstruction is to compactify the space $\mathbf C^2$, in order to include the poles.
We carry this out by embedding $\mathbf C^2$ into $\mathbf C\mathbf P^2$ and explicitly represent the system (\ref{eq:PVlog-system}) in the three affine coordinate charts of $\mathbf C\mathbf P^2$, which are given in Sections \ref{chart01}--\ref{chart03}.
The second step in this process results from the occurence of singularities in the Painlev\'e vector field \eqref{eq:PVlog-system}.
These occur in the $(y_{02},z_{02})$ and $(y_{03},z_{03})$ charts. 
The appearence of these singularities is related to irreducibility of the solutions of Painlev\'e equations as indicated by the following theorem, due to Painlev\'e. 

\begin{theorem}[\cite{Painleve1897}]
If the space of initial values for a differential equation is a compact rational variety, then the equation can be reduced either to a linear differential equation of higher order or to an equation for elliptic functions.
\end{theorem}
It is well known that the solutions of Painlev\'e equations are irreducible (in the sense of the theorem). Since $\mathbf{CP}^2$ is a compact rational variety, the theorem implies $\mathbf{CP}^2$ cannot be the space of initial values for (\ref{eq:PVlog-system}).

By the term {\em singularity} we mean points where $(dy/dt, dz/dt)$ becomes either unbounded or undefined because at least one component approaches the undefined limit $0/0$.
We are led therefore to construct a space in which the points where the singularities are regularised.
The process of regularisation is called "blowing up" or \emph{resolving a singularity}.

We now define the notion of initial value space.

\begin{definition}\label{def:initial-values-space}[\cite{Gerard1975}, \cite{GerardSec1972,Gerard1983,Okamoto1979}]
Let $(\mathcal{E},\pi,\mathcal{B})$ be a complex analytic fibration, $\Phi$ its foliation, and $\Delta$ a a holomorphic differential system on $\mathcal{E}$, such that:
\begin{itemize}
\item the leaves of $\Phi$ correspond to the solutions of $\Delta$;
\item the leaves of $\Phi$ are transversal to the fibres of $\mathcal{E}$;
\item for each path $p$ in the base $\mathcal{B}$ and each point $X\in \mathcal{E}$, such that $\pi(X)\in p$, the path $p$ can be lifted into the leaf of $\Phi$ containing point $X$.
\end{itemize}
Then each fibre of the fibration is called \emph{a space of initial values} for the system $\Delta$.
\end{definition}

The properties listed in Definition \ref{def:initial-values-space} imply that each leaf of the foliation is isomorphic to the base $\mathcal{B}$.
Since the fifth Painlev\'e transcendents (in the $t$ variable) can be globally extended as meromorphic functions on $\mathbf{C}$ \cite{JK1994,HL2001}, we are searching for the fibration with the base equal to $\mathbf{C}$. 




In order to construct the fibration, we apply the blow-up procedure, defined below, \cite{HartshorneAG,GrifHarPRINC,DuistermaatBOOK} to the singularities of the system (\ref{eq:PVlog-system}) that occur where at least one component becomes undefined of the form $0/0$.
\ocite{Okamoto1979} showed that such singular points are contained in the closure of infinitely many leaves.
Moreover, these leaves are holomorphically extended at such a point.

\begin{definition}\label{def:blow-up}
\emph{The blow-up} of the plane $\mathbf{C}^2$ at point $(0,0)$ is the closed subset $X$ of $\mathbf{C}^2\times\mathbf{CP}^1$ defined by the equation $u_1t_2=u_2t_1$, where $(u_1,u_2)\in\mathbf{C}^2$ and $[t_1:t_2]\in\mathbf{CP}^1$, see Figure \ref{fig:blow-up}.
There is a natural morphism $\varphi: X\to\mathbf{C}^2$, which is the restriction of the projection from $\mathbf{C}^2\times\mathbf{CP}^1$ to the first factor.
$\varphi^{-1}(0,0)$ is the projective line $\{(0,0)\}\times\mathbf{CP}^1$, called \emph{the exceptional line}.
\end{definition}
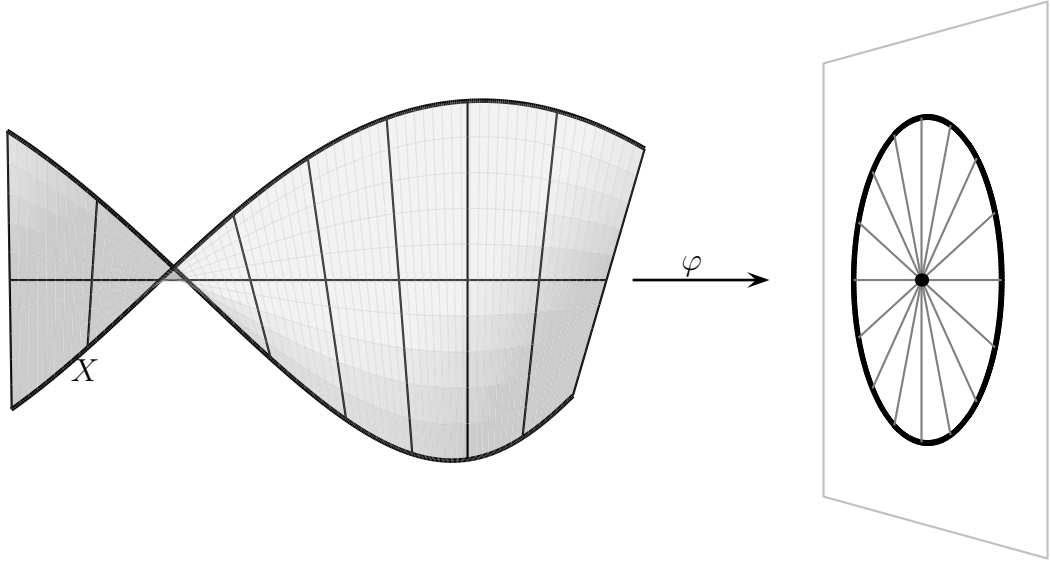
\begin{figure}[h]
\centering
\begin{pspicture*}(-10.5,-4)(4,4)

\psset{unit=0.6}

\psset{viewpoint=4 10 0,Decran=40,lightsrc=20 20 20}

 \psline[linecolor=black,fillcolor=black,incolor=black,linewidth=1pt,arrows=->,arrowsize=2pt 4](-3,0)(0,0)
\rput(-1.7,0.3){$\varphi$}

\rput(-15,-2){$X$}

\psSolid[object=plan, linecolor=gray!50, definition=equation, args={[1 0 0 1]}, base=-1.5 1.5 -1.5 1.5]

\defFunction[algebraic]{krug}(t)
{0}{cos(t)}{sin(t)}
\psSolid[object=courbe,linecolor=black,r=0.01,range=0 pi 2 mul,
linewidth=0.06,resolution=1000,
function=krug](-1,0,0)

\psSolid[object=line, linecolor=gray, args=-1 1 0 -1 -1 0]
\psSolid[object=line, linecolor=gray, args=-1 0 1 -1  0 -1]
\psSolid[object=line, linecolor=gray, args=-1 0.707 0.707 -1  -0.707 -0.707]
\psSolid[object=line, linecolor=gray, args=-1 -0.707 0.707 -1  0.707 -0.707]
\psSolid[object=line, linecolor=gray, args=-1 -0.383 0.924 -1  0.383 -0.924]
\psSolid[object=line, linecolor=gray, args=-1 0.383 0.924 -1  -0.383 -0.924]
\psSolid[object=line, linecolor=gray, args=-1  0.924 -0.383 -1   -0.924 0.383]
\psSolid[object=line, linecolor=gray, args=-1  -0.924 -0.383 -1   0.924 0.383]

\psSolid[object=point, args=-1 0 0]

\defFunction[algebraic]{helix1}(h)
{h+1}{cos(h+pi/4)}{sin(h+pi/4)}
\psSolid[object=courbe,r=0,
range=0 pi, linecolor=black, linewidth=0.1, resolution=360, function=helix1]

\defFunction[algebraic]{helix2}(h)
{h+1}{-cos(h+pi/4)}{-sin(h+pi/4)}
\psSolid[object=courbe,r=0,
range=0 pi, linecolor=black, linewidth=0.1, resolution=360, function=helix2]%

\psSolid[object=line, args=1 0 0 1 pi add 0 0]

\psSolid[object=line, linecolor=black, args=1 0.707 0.707 1 -0.707 -0.707 ]
\psSolid[object=line, linecolor=black, args=2.571 -0.707 0.707 2.571 0.707 -0.707 ]
\psSolid[object=line, linecolor=black, args=2.178 -0.383 0.924 2.178  0.383 -0.924 ]
\psSolid[object=line, linecolor=black, args=1.393 -0.383 -0.924 1.393  0.383 0.924 ]
\psSolid[object=line, linecolor=black, args=2.963  0.924 -0.383 2.963   -0.924 0.383]
\psSolid[object=line, linecolor=black, args=3.75  -0.924 -0.383 3.75  0.924 0.383]
\psSolid[object=line, linecolor=black, args=1.8   0 1 1.8  0 -1 ]
\psSolid[object=line, linecolor=black, args=4.14 0.707 0.707 4.14 -0.707 -0.707 ]

\defFunction[algebraic]{helix}(t,h)
{h+1}{t*cos(h+pi/4)}{t*sin(h+pi/4)}
\psSolid[object=surfaceparametree,linewidth=1sp,linecolor=gray!50,
     base=-1 1 0 pi,fillcolor=gray!50,incolor=gray!50,opacity=0.2,
     function=helix,
   ngrid=10 72]

\end{pspicture*}
\caption{The blow-up of the plane at a point.}\label{fig:blow-up}
\end{figure}

\begin{remark}
Notice that the points of the exceptional line $\varphi^{-1}(0,0)$ are in bijective correspondence with the lines containing $(0,0)$.
On the other hand,
$\varphi$ is an isomorphism between $X\setminus\varphi^{-1}(0,0)$ and $\mathbf{C}^2\setminus\{(0,0)\}$.
More generally, any complex two-dimensional surface can be blown up at a point \cite{HartshorneAG,GrifHarPRINC,DuistermaatBOOK}.
In a local chart around that point, the construction will look the same as described for the case of the plane.
\end{remark}

Notice that the blow-up construction separates the lines containing the point $(0,0)$ in Definition \ref{def:blow-up}, as shown in Figure \ref{fig:blow-up}.
In this way, the solutions of \eqref{eq:PVlog-system} containing the same point can be separated.
Additional blow-ups may be required if the solutions have a commont tangent line or a tangency of higher order at such a point.

The explicit resolution of the vector field \eqref{eq:PVlog-system} is carried out in Appendix \ref{sec:resolution}. As mentioned above, the process requires 11 resolutions of singularities, or, blow-ups.

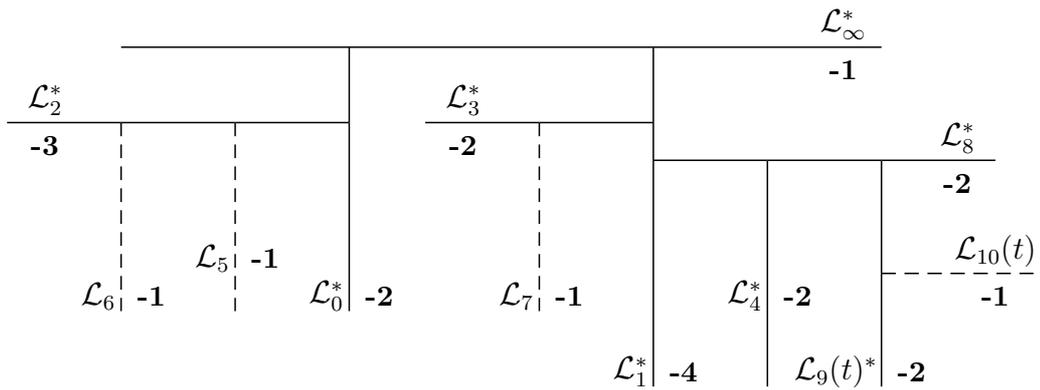
\begin{figure}[h]
\centering
\begin{pspicture}(-1.5,-5)(12,0)

\psset{linecolor=black,linewidth=0.02,fillstyle=solid}

\psline(0,-0.5)(10,-0.5)
\rput(9.5,-0.2){$\mathcal{L}_{\infty}^*$}
\rput(9.5,-0.8){\small\textbf{-1}}

\psline(3,-0.5)(3,-4)
\rput(2.7,-3.8){$\mathcal{L}_{0}^*$}
\rput(3.4,-3.8){$\small\textbf{-2}$}


\psline(3,-1.5)(-1.5,-1.5)
\rput(-1,-1.2){$\mathcal{L}_{2}^*$}
\rput(-1,-1.8){\small\textbf{-3}}


\psline[linestyle=dashed](1.5,-1.5)(1.5,-4)
\rput(1.2,-3.3){$\mathcal{L}_{5}$}
\rput(1.9,-3.3){\small\textbf{-1}}


\psline[linestyle=dashed](0,-1.5)(0,-4)
\rput(-0.3,-3.8){$\mathcal{L}_{6}$}
\rput(0.4,-3.8){\small\textbf{-1}}

\psline(7,-0.5)(7,-5)
\rput(6.7,-4.8){$\mathcal{L}_{1}^*$}
\rput(7.4,-4.8){\small\textbf{-4}}


\psline(7,-1.5)(4,-1.5)
\rput(4.5,-1.2){$\mathcal{L}_{3}^*$}
\rput(4.5,-1.8){\small\textbf{-2}}

\psline[linestyle=dashed](5.5,-1.5)(5.5,-4)
\rput(5.2,-3.8){$\mathcal{L}_{7}$}
\rput(5.9,-3.8){\small\textbf{-1}}

\psline(7,-2)(11.5,-2)
\rput(11,-1.7){$\mathcal{L}_{8}^*$}
\rput(11,-2.3){\small\textbf{-2}}

\psline(8.5,-2)(8.5,-5)
\rput(8.2,-3.8){$\mathcal{L}_{4}^{*}$}
\rput(8.9,-3.8){\small\textbf{-2}}

\psline(10,-2)(10,-5)
\rput(9.4,-4.8){$\mathcal{L}_{9}(t)^{*}$}
\rput(10.4,-4.8){\small\textbf{-2}}

\psline[linestyle=dashed](10,-3.5)(12,-3.5)
\rput(11.5,-3.2){$\mathcal{L}_{10}(t)$}
\rput(11.5,-3.8){\small\textbf{-1}}

\end{pspicture}
\caption{The fibre $\mathcal{D}(t)$ is obtained from $\mathbf{CP}^2$ by $11$ blow-ups.}\label{fig:okamoto}
\end{figure}

The resulting surface $\mathcal{D}(t)$ is shown in Figure \ref{fig:okamoto}.
Note that in this figure, $\mathcal{L}_{\infty}^*$ is the proper preimage of the line at the infinity, while
$\mathcal{L}_{0}^*$, $\mathcal{L}_{1}^*$, $\mathcal{L}_{2}^*$, $\mathcal{L}_{3}^*$, $\mathcal{L}_{4}^*$, $\mathcal{L}_{8}^*$, $\mathcal{L}_{9}(t)^*$  are proper preimages of exceptional lines obtained by blow ups at points $a_0$, $a_1$, $a_2$, $a_3$, $a_4$, $a_8$, $a_9$ respectively
and $\mathcal{L}_{5}$, $\mathcal{L}_{6}$, $\mathcal{L}_{7}$, $\mathcal{L}_{10}(t)$ are exceptional lines obtained by blowing up points $a_5$, $a_6$, $a_7$, $a_{10}$ respectively. 
The self-intersection number of each line (after all blow-ups are completed) is indicated in the figure.

Okamoto described so called \emph{singular points of the first class} that are not contained in the closure of any leaf of the foliation given by the system of differential equations.
At such points, the corresponding vector field is infinite.
(For example, in chart $(y_{02},z_{02})$ from Section \ref{chart02} such a singular point is given by $y_{02}=0$ with non-zero $z_{02}$.)
In the surface $\mathcal{D}(t)$, all remaining singular points are of the first class, and the fibre of the initial value space is obtained by removing them:
$$
\mathcal{E}(t)=\mathcal{D}(t)\setminus\left(\bigcup_{j=0}^4\mathcal{L}_j^*\cup\mathcal{L}_8^*\cup\mathcal{L}_9^*\cup\mathcal{L}_{\infty}\right).
$$

In $\mathcal{D}(t)$, each line with self-intersection number $-1$ can be blown down again. 
Blowing down $\mathcal{L}_{\infty}^*$ and then the projection of $\mathcal{L}_{0}^*$, we get the surface $\mathcal{F}(t)$, which is shown in Figure \ref{fig:okamoto-blow-down}.
The projection of each remaining line from $\mathcal{D}(t)$ is denoted by the same index but now with superscript $p$.
Notice that the self-intersection numbers of $\mathcal{L}_{1}^p$ and $\mathcal{L}_{2}^p$ are no longer the same as of the corresponding pre-images $\mathcal{L}_{1}^*$ and $\mathcal{L}_{2}^*$.
In this space,  we denote by $\mathcal{I}$ the set of all singular points of the first class in $\mathcal{F}(t)$:
$$
\mathcal{I}=\bigcup_{j=1}^4\mathcal{L}_j^p\cup\mathcal{L}_8^p\cup\mathcal{L}_9^p.
$$
The fibre $\mathcal{E}(t)$ of the initial value space can be identified with $\mathcal{F}(t)\setminus\mathcal{I}$.

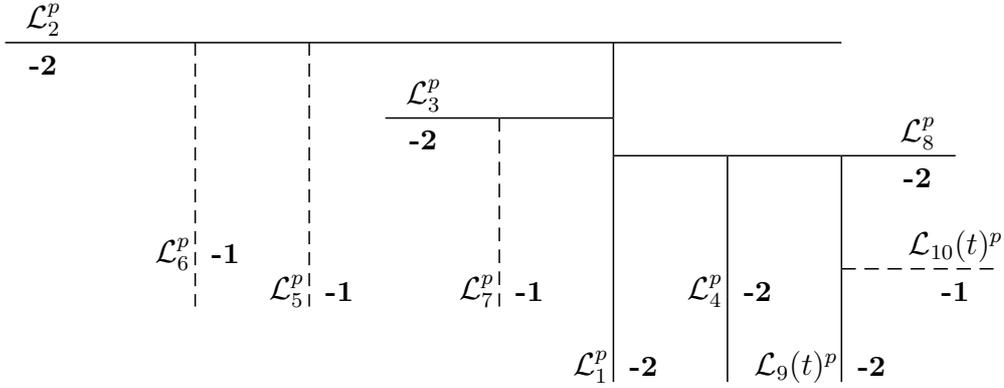
\begin{figure}[h]
\centering
\begin{pspicture}(-1,-5)(12,0)

\psset{linecolor=black,linewidth=0.02,fillstyle=solid}

\psline(-1,-0.5)(10,-0.5)
\rput(-.5,-0.2){$\mathcal{L}_{2}^p$}
\rput(-.5,-0.8){\small\textbf{-2}}

\psline[linestyle=dashed](3,-0.5)(3,-4)
\rput(2.7,-3.8){$\mathcal{L}_{5}^p$}
\rput(3.4,-3.8){\small\textbf{-1}}


\psline[linestyle=dashed](1.5,-0.5)(1.5,-4)
\rput(1.2,-3.3){$\mathcal{L}_{6}^p$}
\rput(1.9,-3.3){\small\textbf{-1}}

\psline(7,-0.5)(7,-5)
\rput(6.7,-4.8){$\mathcal{L}_{1}^p$}
\rput(7.4,-4.8){\small\textbf{-2}}


\psline(7,-1.5)(4,-1.5)
\rput(4.5,-1.2){$\mathcal{L}_{3}^p$}
\rput(4.5,-1.8){\small\textbf{-2}}

\psline[linestyle=dashed](5.5,-1.5)(5.5,-4)
\rput(5.2,-3.8){$\mathcal{L}_{7}^p$}
\rput(5.9,-3.8){\small\textbf{-1}}

\psline(7,-2)(11.5,-2)
\rput(11,-1.7){$\mathcal{L}_{8}^p$}
\rput(11,-2.3){\small\textbf{-2}}

\psline(8.5,-2)(8.5,-5)
\rput(8.2,-3.8){$\mathcal{L}_{4}^p$}
\rput(8.9,-3.8){\small\textbf{-2}}

\psline(10,-2)(10,-5)
\rput(9.4,-4.8){$\mathcal{L}_{9}(t)^p$}
\rput(10.4,-4.8){\small\textbf{-2}}

\psline[linestyle=dashed](10,-3.5)(12,-3.5)
\rput(11.5,-3.2){$\mathcal{L}_{10}(t)^p$}
\rput(11.5,-3.8){\small\textbf{-1}}

\end{pspicture}
\caption{The fibre $\mathcal{F}(t)$ is obtained from $\mathbf{CP}^2$ by $11$ blow-ups and two blow-downs.}\label{fig:okamoto-blow-down}
\end{figure}

If $\mathcal{S}$ is a surface obtained from the projective plane by a several successive blow-ups of points, then the group of all automorphisms of the Picard group $\Pic(\mathcal{S})$ preserving the canonical divisor $K$ is generated by the reflections
$X\mapsto X+(X.\omega)\omega$,
with $\omega\in\Pic(\mathcal{S})$, $K.\omega=0$, $\omega.\omega=-2$.
That group is an affine Weyl group and the lines of self-intersection $-2$ are its simple roots.
Representing each such line by a node and connecting a pair of nodes by a line only if they intersect in the fibre, we obtain the Dynkin diagram shown in Figure \ref{fig:dynkin}, which is of type $D_5^{(1)}$.
For detailed expositions on the topic of surfaces and root systems, see \cite{Demazure1980, Harbourne1985} and references therein.

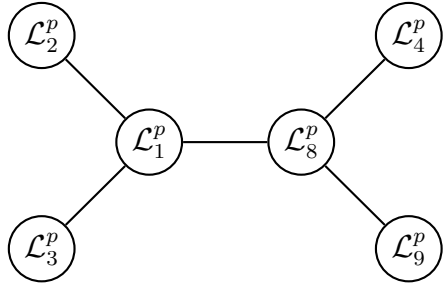
\begin{figure}[h]
\centering
\begin{pspicture}(-2,-2)(4,2)

\psset{linecolor=black}
\rput(0,0){\circlenode{1}{$\mathcal{L}_1^p$}}
\rput(2,0){\circlenode{2}{$\mathcal{L}_8^p$}}
\rput(-1.414,1.414){\circlenode{3}{$\mathcal{L}_2^p$}}
\rput(-1.414,-1.414){\circlenode{4}{$\mathcal{L}_3^p$}}
\rput(3.414,1.414){\circlenode{5}{$\mathcal{L}_4^p$}}
\rput(3.414,-1.414){\circlenode{6}{$\mathcal{L}_9^p$}}

\ncline{-}{1}{2}
\ncline{-}{1}{3}
\ncline{-}{1}{4}
\ncline{-}{2}{5}
\ncline{-}{2}{6}
    
\end{pspicture}
\caption{The Dynkin diagram of $D_5^{(1)}$.}\label{fig:dynkin}
\end{figure}

In the limit $\Real t\to-\infty$, the resulting Okamoto space is compactified by the fibre $\mathcal{F}_{\infty}$, corresponding to the autonomous system (\ref{eq:PVlog-system-auto}), see Figure \ref{fig:okamoto-blow-down-limit}.
Its infinity set is given by
$$
\mathcal{I}_{\infty}=\mathcal{L}_1^a\cup\mathcal{L}_2^p\cup\mathcal{L}_3^a.
$$
\begin{figure}[h]
\centering
\begin{pspicture}(-1,-5)(8,0)

\psset{linecolor=black,linewidth=0.02,fillstyle=solid}

\psline(-1,-0.5)(8,-0.5)
\rput(-.5,-0.2){$\mathcal{L}_{2}^a$}
\rput(-.5,-0.8){\small\textbf{-2}}

\psline[linestyle=dashed](3,-0.5)(3,-4)
\rput(2.7,-3.8){$\mathcal{L}_{5}^a$}
\rput(3.4,-3.8){\small\textbf{-1}}


\psline[linestyle=dashed](1.5,-0.5)(1.5,-4)
\rput(1.2,-3.3){$\mathcal{L}_{6}^a$}
\rput(1.9,-3.3){\small\textbf{-1}}

\psline(7,-0.5)(7,-5)
\rput(6.7,-4.8){$\mathcal{L}_{1}^a$}
\rput(7.4,-4.8){\small\textbf{-1}}


\psline(7,-1.5)(4,-1.5)
\rput(4.5,-1.2){$\mathcal{L}_{3}^a$}
\rput(4.5,-1.8){\small\textbf{-2}}

\psline[linestyle=dashed](5.5,-1.5)(5.5,-4)
\rput(5.2,-3.8){$\mathcal{L}_{7}^a$}
\rput(5.9,-3.8){\small\textbf{-1}}

\end{pspicture}
\caption{The fibre $\mathcal{F}_{\infty}$ corresponding to the autonomous system (\ref{eq:PVlog-system-auto}).}\label{fig:okamoto-blow-down-limit}
\end{figure}
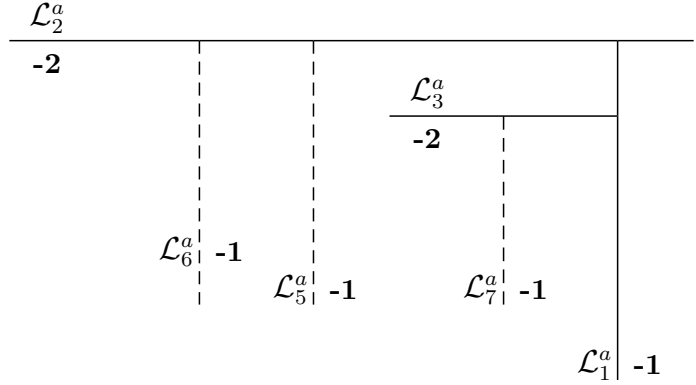
$\mathcal{F}_{\infty}$ is constructed by a sequence of blow-ups, carried out along the lines of the construction of the space of the initial values for the non-autonomous system (\ref{eq:PVlog-system}).
For those readers who may wish to carry this out separately, we note that the resolution at point $a_4$ (found in Section \ref{a1-blow}) needs to be done only for the non-autonomous system.

\subsection{Movable singularities in Okamoto's space}\label{sec:poles}
Here, we consider neighbourhoods of exceptional lines where the Painlev\'e vector field \eqref{eq:PVlog-system} becomes unbounded.
The construction given in Appendix \ref{sec:resolution} shows that these are given by the lines $\mathcal{L}_5$, $\mathcal{L}_6$, $\mathcal{L}_7$ and $\mathcal{L}_{10}$.

\subsubsection*{Movable pole with residue $-\,\theta_{\infty}$} The set $\mathcal{L}_5\setminus\mathcal{I}$ is given by $y_{62}=0$ in the $(y_{62},z_{62})$ chart, see Section \ref{a5-blow}.
Suppose $y_{62}(\tau)=0$, $z_{62}(\tau)=B$, for arbitrary complex numbers $\tau$, $B$.
Solving the system of differential equations for $y_{62}$, $z_{62}$ in Section \ref{a5-blow} in a neighbourhood of $t=\tau$, we obtain
$$
y_{62}=-\theta_{\infty}(t-\tau)+\frac{\theta_{\infty}}{2}\big(2 B-2\theta_{\infty}-\eta-\theta_1e^{\tau}\big)(t-\tau)^2+O((t-\tau)^3).
$$
Since the transformation from $(y_{62}, z_{62})$ to $(y, z)$ (also given in Section \ref{a5-blow}) is
$$
y=\frac1{y_{62}}
\quad\text{and}\quad
z=y_{62}(y_{62}z_{62}+1),
$$ 
we obtain series expansions for $(y, z)$ given by
$$
\begin{aligned}
y= & -\frac{\theta_{\infty}}{t-\tau}-\frac{2 B-2\theta_{\infty}-\eta-\theta_1e^{\tau}}{2\theta_{\infty}}+O(t-\tau),
\\
z= & -\theta_{\infty}(t-\tau)
 +
\frac{\theta_{\infty}}{2}\left(
2B(1+\theta_{\infty})-\eta-2\theta_{\infty}-\theta_1 e^{\tau}
\right)(t-\tau)^2
+O((t-\tau)^3).
\end{aligned}
$$
Clearly, $y$ has a simple pole with residue $-\theta_{\infty}$,  while $z$ has a simple zero at $t=\tau$.

\subsubsection*{Movable pole with residue $\theta_{\infty}$} At the intersection with $\mathcal{L}_6\setminus\mathcal{I}$, $y$ has a simple pole with the residue $\theta_{\infty}$, while $z$ has a simple zero.
This case is analogous to the previous one, see Sections \ref{a5-blow} and \ref{a6-blow}.

\subsubsection*{Movable zero with coefficient $\theta_{0}$} The set $\mathcal{L}_7\setminus\mathcal{I}$ is given by $z_{81}=0$ in the $(y_{81},z_{81})$ chart, see Section \ref{a7-blow}.
Suppose $y_{81}(\tau)=B$ and $z_{81}(\tau)=0$.
Then integration of the vector field gives
\[
z_{81}(t)=(t-\tau)
+\frac{\eta-2\theta_0-\theta_1e^{\tau}}{2}(t-\tau)^2
+\frac{F-2\eta\theta_0+\theta_0^2-\theta_1e^{\tau}-4B}{3}(t-\tau)^3
 +O((t-\tau)^4),
\]
with $F=\dfrac12\epsilon(\theta_0+\eta-\theta_{\infty})$.
Since
$$
y=z_{81}(y_{81}z_{81}+\theta_0),
\quad
z=\frac1{z_{81}},
$$
we obtain
\[
\begin{split}
y= &\ \theta_0(t-\tau)+\frac{2B+\theta_0(\eta-2\theta_0-\theta_1e^{\tau})}{2}(t-\tau)^2+O((t-\tau)^3),
\\
z=&\ \frac{1}{t-\tau}-\frac{\eta-2\theta_0-\theta_1e^{\tau}}{2}
+\left(\frac{(\eta-2\theta_0-\theta_1e^{\tau})^2}{4}-\frac{F-2\eta\theta_0+\theta_0^2-\theta_1e^{\tau}-4B}{3}\right)(t-\tau)
\\&\qquad\quad
+O((t-\tau)^2).
\end{split}
\]
Thus, $y$ has a simple zero and $z$ a simple pole at $t=\tau$.

\subsubsection*{Movable points where $y$ becomes unity} The set $\mathcal{L}_{10}\setminus\mathcal{I}$ is given by $z_{111}=0$ in the $(y_{111},z_{111})$ chart, see Section \ref{a10-blow}.
Suppose $y_{111}(\tau)=B$ and $z_{111}(\tau)=0$.
Then
\[
\begin{split}
z_{111}=&\ (t-\tau)+\frac{\theta_1 e^{\tau}-\eta-1}{2}(t-\tau)^2
\\&
+\frac13\left(
1+\frac52\eta+\eta^2-\frac{B}{\theta_1 e^{\tau}}+\theta_1e^{\tau}\left(1-\theta_0-\frac{\eta}2\right)+\frac12\theta_1^2e^{2\tau}
\right)(t-\tau)^3
 +O((t-\tau)^4).
\end{split}
\]
Since we have
$$
\begin{aligned}
y =&\ 1 + \theta_1 e^t z_{111} + (1+\eta)\theta_1 e^t z_{111}^2 + y_{111} z_{111}^3,
\\
z =&\ \frac{1}{ z_{111}^2 (\theta_1 e^t + (1+\eta)\theta_1 e^t z_{111} + y_{111} z_{111}^2)},
\end{aligned}
$$
it follows that $y(\tau)=1$ and $z$ has a double pole at $t=\tau$.
Their expansions around this point are given by
\[
\begin{split}
y=&\ 1+\theta_1e^{\tau}(t-\tau)+\frac{\theta_1e^{\tau}(\theta_1 e^{\tau}+\eta+3)}{2}(t-\tau)^2
\\
&+\frac{\theta_1e^{\tau}}{6}\left(
2 - 4 \eta - 4 \eta^2 - \frac{2 B}{\theta_1 e^{\tau}} + (11 +  5 \eta- 2 \theta_0)\theta_1 e^{\tau}
 + \theta_1^2e^{2\tau}
\right)(t-\tau)^3
+O((t-\tau)^4),
\\
z=&\ \frac{(\theta_1 e^{\tau})^{-1}}{(t-\tau)^2}+\frac{1+(\theta_1e^{\tau})^{-1}}{t-\tau}
+ \left(\frac{2 B}{\theta_1 e^{\tau}} + \frac{\eta \theta_1 e^{\tau}}{2}+ \theta_0 \theta_1 e^{\tau} 
+ \frac{\theta_1^2 e^{2\tau}}{4} - \frac{1+3 \eta^2}{4} \right)
+O(1).
\end{split}
\]

\section{Special solutions}
\label{sec:special}
In this section, we consider the pencil of curves corresponding to the Hamiltonian (\ref{eq:E}) of the autonomous system (\ref{eq:PVlog-system-auto}).
We show that by a birational equivalence, this pencil can be transformed to a pencil of conics.
In Section \ref{sec:conics}, we analyse the solutions of the fifth Painlev\'e equation corresponding to singular conics from the pencil.
In Section \ref{sec:other}, we give a summary of other special solutions.

\subsection{Special solutions and singular conics}
\label{sec:conics}

The pencil of curves arising from the Hamiltonian  of the autonomous system (\ref{eq:PVlog-system-auto}) is given by the zero set of the one-parameter family of polynomials (parametrised by $c$):
\begin{equation}\label{eq:pencil}
h_c(y,z)=y(y-1)^2z^2-(\theta_0+\eta)y^2z+(2\theta_0+\eta)yz-\theta_0z+\frac12\epsilon(\theta_0+\eta-\theta_{\infty})y-c.
\end{equation}
For each $c$, the curve $h_c(y,z)=0$ is birationally equivalent to
$$
h_{\tilde c}^1(y_1,z_1)=\frac{z_1^2-\left( \theta_0^2 +\tilde{c} y_1  + \theta_{\infty}^2  y_1^2 \right)}{4y_1}=0,
$$
where $\tilde{c}=4c-2\theta_0(\eta+\theta_0)$ and the birational equivalence is given by
$$
y_1=y,
\quad
z_1=2y(y-1)z-(\theta_0+\eta)y+\theta_0.
$$
The family of level curves $h_{\tilde c}^1=0$ forms a pencil of conics.
The singular conics in the pencil are:
\begin{align*}
\tilde{c}&=&2\theta_0\theta_{\infty},
&&(z_1+\theta_{\infty} y_1+\theta_0)(z_1-\theta_{\infty}y_1-\theta_0)=0,
\\
\tilde{c}&=&-2\theta_0\theta_{\infty},
&&(z_1-\theta_{\infty} y_1+\theta_0)(z_1+\theta_{\infty}y_1-\theta_0)=0,
\\
\tilde{c}&=&\infty,
&&
y_1w_1=0,
\end{align*}
where $[y_1:z_1:w_1]$ are homogeneous coordinates.
Moreover, the base points of this pencil are:
$$
[0:\theta_0:1],
\quad
[0:-\theta_0:1],
\quad
[1:\theta_{\infty}:0],
\quad
[1:-\theta_{\infty}:0].
$$

Recalling that the Hamiltonian $E$ of the autonomous system, given by Equation \eqref{eq:E}, has the time derivative 
$$
E'=-\theta_1 e^t\big( z(2\theta_0-(\eta+2\theta_0)y+2y(y-1)z ) +E\big),
$$
we can also search for conditions under which all successive derivatives of $E$ are zero. This happens if and only if
$$
z=0
\quad\text{or}\quad
2\theta_0-(\eta+2\theta_0)y+2y(y-1)z =0.
$$
These cases are investigated in further detail below.

\subsubsection*{Case $z=0$}
From the second equation of (\ref{eq:PVlog-system}) we have $\epsilon(\theta_0+\eta-\theta_{\infty})=0$, i.e.~ $\theta_0+\eta=\pm\theta_{\infty}$.
In this case $z_1=\pm\theta_{\infty}y_1+\theta_0$, and so this case corresponds to lines in the pencil of conics containing the base point $[0:\theta_0:1]$.

The first equation of (\ref{eq:PVlog-system}) is then a Riccati equation:
$$
\frac{dy}{dt}=\mp\theta_{\infty} y^2+(\theta_0\pm\theta_{\infty}-\theta_1 e^t)y-\theta_0,
$$
which is equivalent to 
\begin{equation}\label{eq:riccati1}
x\frac{dy}{dx}=\mp\theta_{\infty} y^2+(\theta_0\pm\theta_{\infty}-\theta_1 x)y-\theta_0.
\end{equation}
For $\theta_{\infty}\theta_1\neq0$, the solutions of this equation can be expressed in terms of the Whittaker functions \cite{NISThandbook}.
Note that $\theta_1\neq0$ is equivalent to $\delta\neq0$, when the $\PV$ can be renormalised to $\delta=-1/2$.
Then, the solutions of (\ref{eq:riccati1}) are given by
$$
y=-\frac{z\phi'(x)}{\theta_{\infty}\phi(x)},
$$
where
$$
\phi(x)=\frac{C_1M_{\kappa,\mu}(\xi)+C_2W_{\kappa,\mu}(\xi)}{\xi^{\kappa}}\, e^{\xi/2},
$$
with $\xi=\pm x$, $\kappa=(\mp\theta_{\infty}-\theta_0+1)/2$, $\mu=\mp\theta_{\infty}+\theta_0$, and $C_1$, $C_2$ being arbitrary constants.

Solutions from this class intersect only the pole lines $\mathcal{L}_5$ and $\mathcal{L}_6$.

\subsubsection*{Case $2\theta_0-(\eta+2\theta_0)y+2y(y-1)z =0$}
The first equation of (\ref{eq:PVlog-system}) is then:
$$
\frac{dy}{dt}=\theta_0 y^2-(2\theta_0+\theta_1 e^t)y+\theta_0,
$$
which is again a Riccati equation that can be solved analogously to the previous case of (\ref{eq:riccati1}).
Since in this case $z_1=\theta_0 y-\theta_0$, the condition on the constants is $\theta_0=\pm\theta_{\infty}$.
The solutions belong to the lines from the pencil of conics that contain the base point $[0:-\theta_0:1]$.

Solutions from this class intersect only the pole line $\mathcal{L}_7$.

\subsection{Other special solutions of $\PV$}
\label{sec:other}
In this section, we give a brief account on special solutions and B\"acklund transformations of $\PV$.
For a detailed exposition, we refer readers to 
\cite{Luk1968, Airault1979, Okamoto1987, Gromak1999, Schief2000,  FW2002, MOK2002, Clark2005b, NISThandbook} and the references therein.

All rational solutions of $\PV$ are of the form
$$
y=\lambda x+\mu+\frac{P(x)}{Q(x)},
$$
where $\lambda$, $\mu$ are constants and $P$, $Q$ polynomials of degrees $n-1$ and $n$ respectively.

Such solutions can be obtained by applying B\"acklund transformations:
$$
\begin{aligned}
& \mathcal{S}_1\ :\ (y(x);\alpha,\beta,\gamma,\delta)\mapsto(y(-x);\alpha,\beta,-\gamma,\delta),
\\
&\mathcal{S}_2\ :\ (y(x);\alpha,\beta,\gamma,\delta)\mapsto\left(\frac{1}{y(x)};-\beta,-\alpha,-\gamma,\delta\right),
\\
&\mathcal{T}_{\varepsilon_1,\varepsilon_2,\varepsilon_3}\ :\ \left(y(x);\alpha,\beta,\gamma,-\frac12\right)
\mapsto
\left(\frac{\Phi-2\varepsilon_1 xy}{\Phi};\alpha',\beta',\gamma',-\frac12\right),
\end{aligned}
$$
where $\varepsilon_1, \varepsilon_2, \varepsilon_3\in\{-1,1\}$ and
$$
\begin{aligned}
&\alpha'=\frac18\left(\gamma+\varepsilon_1\left(1-\varepsilon_3\sqrt{-2\beta}-\varepsilon_2\sqrt{2\alpha}\right)\right)^2;
\\
&\beta'=-\frac18\left(\gamma-\varepsilon_1\left(1-\varepsilon_3\sqrt{-2\beta}-\varepsilon_2\sqrt{2\alpha}\right)\right)^2;
\\
&\gamma'=\varepsilon_1\left(\varepsilon_3\sqrt{-2\beta}-\varepsilon_2\sqrt{2\alpha}\right);
\\
&\Phi=xy'-\varepsilon_2\sqrt{2\alpha} y^2+\varepsilon_3\sqrt{-2\beta}
+
(\varepsilon_2\sqrt{2\alpha}-\varepsilon_3\sqrt{-2\beta}+\varepsilon_1x)y.
\end{aligned}
$$
to so called seed solutions of $\PV$, which are given by
\begin{equation}\label{eq:rational}
y=\begin{cases}
&\kappa x+\mu, \ \text{for}\ \alpha=\frac12,\ \beta=-\frac12\mu^2,\ \gamma=\kappa(2-\mu),\ \delta=-\frac12\kappa^2;
\\
&\dfrac{\kappa}{x+\kappa}, \  \ \text{for}\ \alpha=\frac12,\ \beta=\kappa^2\mu,\ \gamma=2\kappa\mu,\ \delta=\mu;
\\
&\dfrac{\kappa+x}{\kappa-x}, \ \ \text{for} \ \alpha=\frac18,\ \beta=-\frac18,\ \gamma=-\kappa\mu,\ \delta=\mu,
\end{cases}
\end{equation}
where $\kappa$, $\mu$ are arbitrary constants.

In addition, $\PV$ has the following elementary (non-rational) solutions:
\begin{equation}\label{eq:elementary}
\begin{aligned}
&y=1+\kappa\sqrt{x},\quad \text{for} \quad\alpha=\mu,\ \beta=-\frac18,\ \gamma=-\mu\kappa^2,\ \delta=0;
\\
&y=\kappa e^{\mu x},\quad \text{for} \quad\alpha=\beta=0,\ \gamma=\mu,\ \delta=-\frac12\mu^2,
\end{aligned}
\end{equation}
where $\kappa$, $\mu$ are arbitrary constants.
For $\delta=0$, algebraic solutions are all rational functions in $\sqrt{x}$ and can be obtained by consecutive application of the B\"acklund transformations $\mathcal{S}_1$ and $\mathcal{S}_2$ to the first function of (\ref{eq:elementary}).

\section{The solutions near the infinity set}
\label{sec:infinity}
In this section, we study the behaviour of the solutions of the system (\ref{eq:PVlog-system}) near the set $\mathcal{I}$, where the vector field is infinite.
We prove that $\mathcal{I}$ is a repeller for the solutions and that each solution which comes sufficiently close to $\mathcal{I}$ at a certain point $t$ will have a pole in a neighbourhood of $t$.

\begin{lemma}\label{lemma:L2}
For every $\epsilon_1>0$ there exists a neighbourhood $U$ of $\mathcal{L}_2^p$ such that
$$
\left|
\frac{E'}{E}+\theta_1 e^t
\right|
<\epsilon_1
\quad
\text{in}\ U.
$$
\end{lemma}

\begin{proof}
In the corresponding charts near $\mathcal{L}_2^p$ (see Section \ref{a2-blow}), the function
$$
r=\frac{E'}{E}+\theta_1 e^t
$$
is equal to:
$$
\begin{aligned}
r_{31} &= \frac{\theta_1 e^t y_{31} z_{31} (\eta y_{31}-2 (\theta_0 y_{31}-1) (y_{31} z_{31}-1))}
{\eta y_{31} (y_{31} z_{31}-1)+F y_{31}^2-\theta_0 y_{31} (y_{31} z_{31}-1)^2+y_{31}^2 z_{31}^2-2 y_{31} z_{31}+1},
\\
r_{32} &= \frac{\theta_1 e^t y_{32} z_{32} (\eta-2 (y_{32}-1) (\theta_0-z_{32}))}
{(y_{32}-1) z_{32} (\eta-(y_{32}-1) (\theta_0-z_{32}))+F}.
\end{aligned}
$$
The statement of the lemma follows immediately from these expressions, since  $\mathcal{L}_2^p$ is given by $z_{31}=0$ and $y_{32}=0$.
\end{proof}

\begin{lemma}\label{lemma:L1}
For each compact subset $K$ of
$\mathcal{L}_1^p\setminus\mathcal{L}_8^p$
there exists a neighbourhood $V$ of $K$ and a constant $C>0$ such that:
$$
\left| e^{-t}\frac{E'}{E} \right|<C
\quad
\text{in}\ V\ \text{for all}\ t.
$$
\end{lemma}

\begin{proof}
Near $\mathcal{L}_1^p$, in the respective coordinate charts (see Section \ref{a1-blow}), we have:
$$
e^{-t}\frac{E'}{E}\sim
\begin{cases}
-\dfrac{\theta_1 (y_{21}+1)}{y_{21}-1},
\\
\dfrac{ \theta_1 (z_{22}+1)}{z_{22}-1}.
\end{cases}
$$
Since the projection of $\mathcal{L}_8^p$ to these two charts is the point on $\mathcal{L}_1^p$ given by coordinates $y_{21}=1$ and $z_{22}=0$, the statement is proved.
\end{proof}

The modulus $|d|$ of the function $d$ from the next lemma will serve as a measure for the distance from the infinity set $\mathcal{I}$.
Throughout the paper, we denote by $J_{n}$ the Jacobian of the coordinate change from $(y,z)$ to $(y_{n},z_{n})$:
$$
J_{n}=\frac{\partial y_{n}}{\partial y}\frac{\partial z_{n}}{\partial z}-\frac{\partial y_{n}}{\partial z}\frac{\partial z_{n}}{\partial y}.
$$

\begin{lemma}
There exists a continuous complex valued function $d$ on a neighbourhood of the infinity set $\mathcal{I}$ in the Okamoto space, such that:
$$
d=
\begin{cases}
\frac1E, & \text{in a neighbourhood of}\ \ \mathcal{I}\setminus(\mathcal{L}_3^p\cup\mathcal{L}_4^p\cup\mathcal{L}_8^p\cup\mathcal{L}_9^p),
\\
-J_{82}, & \text{in a neighbourhood of}\ \ \mathcal{L}_3^p\setminus\mathcal{L}_1^p,
\\
-J_{102}, & \text{in a neighbourhood of}\ \ (\mathcal{L}_4^p\cup\mathcal{L}_8^p)\setminus\mathcal{L}_1^p,
\\
-J_{112}, & \text{in a neighbourhood of}\ \ \mathcal{L}_9^p\setminus\mathcal{L}_8^p.
\end{cases}
$$
\end{lemma}

\begin{proof}


From Section \ref{a7-blow}, the line $\mathcal{L}_{3}^p$ is given by $z_{82}=0$ in the $(y_{82},z_{82})$ chart.
Thus as we approach $\mathcal{L}_{3}^p$, i.e., as $z_{82}\to0$, we have
$$
EJ_{82}\sim-1.
$$


From Section \ref{a9-blow}, the line $\mathcal{L}_{8}^p$ is given by $z_{102}=0$ in the $(y_{102},z_{102})$ chart.
Thus as we approach $\mathcal{L}_{8}^p$, i.e., as $z_{102}\to0$, we have
$$
EJ_{102}\sim-1-\frac{\theta_1 e^t}{y_{102}}.
$$
In the same chart, the line $\mathcal{L}_{4}^p$ is given by $y_{102}=-\theta_1e^t$.

From Section \ref{a10-blow}, the line $\mathcal{L}_9^p$ is given by $z_{112}=0$ in the  $(y_{112},z_{112})$ chart.
We have:
$$
\frac{J_{112}}{J_{102}}=1+\frac{(1+\eta)\theta_1 e^t}{y_{112}}.
$$
 \end{proof}

 \begin{lemma}[Behaviour near $\mathcal{L}_3^p\setminus\mathcal{L}_1^p$]\label{lemma:L3}
 If a solution at a complex time $t$ is sufficiently close to  $\mathcal{L}_3^p\setminus\mathcal{L}_1^p$, then there exists unique $\tau\in\mathbf{C}$ such that $(y(\tau),z(\tau))$ belongs to the line $\mathcal{L}_7$.
  In other words, the solution has a pole at $t=\tau$.
  
Moreover
$|t-\tau|=O(|d(t)||y_{82}(t)|)$
 for sufficiently small $d(t)$ and bounded $|y_{82}|$.
 
 For large $R_3>0$, consider the set $\{ t\in\mathbf{C} \mid |y_{82}|\le R_3\}$.
 Its connected component containing $\tau$ is an approximate disk $D_3$ with centre $\tau$ and radius $|d(\tau)|R_3$,
 and $t\mapsto y_{82}(t)$ is a complex analytic diffeomorphism from that approximate disk onto $\{y\in\mathbf{C}\mid|y|\le R_3\}$.
 \end{lemma}
 
 \begin{proof}
For the study of the solutions near $\mathcal{L}_3^p\setminus\mathcal{L}_1^p$, we use coordinates $(y_{82},z_{82})$, see Section \ref{a7-blow}.
In this chart, the line $\mathcal{L}_3^p\setminus\mathcal{L}_1^p$ is given by the equation $z_{82}=0$ and parametrised by $y_{82}\in\mathbf{C}$.
Moreover, $\mathcal{L}_7^p$ is given by $y_{82}=0$ and parametrised by $z_{82}\in\mathbf{C}$.
 
Asymptotically, for $z_{82}\to0$ and bounded $y_{82}$, $e^{-t}$, we have:
\begin{subequations}
\begin{align}
y_{82}' &\sim\frac1{z_{82}},\label{eq:y82'}
\\
z_{82}' &\sim4y_{82}^3z_{82}^2,\label{eq:z82'}
\\
J_{82} &=-z_{82},\label{eq:J82}
\\
\frac{J'_{82}}{J_{82}} & = \eta-2\theta_0-\theta_1e^t-4y_{82}+O(z_{82})=\eta-2\theta_0-\theta_1e^t-4y_{82}+O(J_{82}),\label{eq:J'/J82}
\\
EJ_{82} & \sim -1.\label{eq:EJ82}
\end{align}
\end{subequations}

Integrating (\ref{eq:J'/J82}) from $\tau$ to $t$, we get
\begin{gather*}
J_{82}(t)=J_{82}(\tau)e^{K(t-\tau)}e^{-\theta_1(e^t-e^{\tau})}(1+o(1)),
\\
K=\eta-2\theta_0-4y_{82}(\tilde{\tau}),
\end{gather*}
where $\tilde{\tau}$ is on the integration path.

Because of (\ref{eq:z82'}), $z_{82}$ is approximately equal to a small constant, and from (\ref{eq:y82'}) follows that:
$$
y_{82}\sim y_{82}(\tau)+\frac{t-\tau}{z_{82}}.
$$
Thus, if $t$ runs over an approximate disk $D$ centered at $\tau$ with radius $|z_{82}|R$, then $y_{82}$ fills and approximate disk centered at $y_{82}(\tau)$ with radius $R$.
Therefore, if $z_{82}(\tau)\ll\tau$, the solution has the following properties for $t\in D$:
$$
\frac{z_{82}(t)}{z_{82}(\tau)}\sim1,
 $$ 
 and $y_{82}$ is a complex analytic diffeomorphism from $D$ onto an approximate disk with centre $y_{82}(\tau)$ and radius $R$.
 If $R$ is sufficiently large, we will have $0\in y_{82}(D)$, i.e.~the solution of the Painlev\'e equation will have a pole at a unique point in $D$.
 
 Now, it is possible to take $\tau$ to be the pole point.
 For $|t-\tau|\ll|\tau|$, we have:
 \begin{gather*}
 \frac{d(t)}{d(\tau)}\sim1,
 \quad\text{i.e.}\quad
 \frac{z_{82}(t)}{d(\tau)}\sim-\frac{J_{82}(t)}{d(\tau)}\sim1,
 \\
 y_{82}(t)\sim\frac{t-\tau}{z_{82}}\sim\frac{t-\tau}{d(\tau)}.
 \end{gather*}
 Let $R_3$ be a large positive real number.
 Then the equation $|y_{82}(t)|=R_3$ corresponds to $|t-\tau|\sim|d(\tau)|R_3$, which is still small compared to $|\tau|$ if $|d(\tau)|$ is sufficiently small.
 Denote by $D_3$ the connected component of the set of all $t\in\mathbf{C}$ such that $\{t\mid |y_{82}(t)|\le R_3\}$ is an approximate disk with centre $\tau$ and radius $2|d(\tau)|R_3$.
More precisely, $y_{82}$ is a complex analytic diffeomorphism from $D_3$ onto $\{y\in\mathbf{C}\mid|y|\le R_3\}$, and
$$
\frac{d(t)}{d(\tau)}\sim1
\quad\text{for all}\quad
t\in D_3.
$$
We have $E(t)J_{82}(t)\sim-1$ when $|z_{82}|\ll1$.
Thus $E(t)J_{82}(t)\sim-1$ for the annular disk $t\in D_3\setminus D_3'$, where $D_3'$ is a disk centered at $\tau$ with small radius compared to radius of $D_3$.
\end{proof}

\begin{lemma}[Behaviour near $\mathcal{L}_9^p\setminus\mathcal{L}_8^p$]\label{lemma:L9}
 If a solution at a complex time $t$ is sufficiently close to  $\mathcal{L}_9^p\setminus\mathcal{L}_8^p$, then there exists unique $\tau\in\mathbf{C}$ such that $(y(\tau),z(\tau))$ belongs to the line $\mathcal{L}_{10}$.
  In other words, the solution has a pole at $t=\tau$.
  
Moreover
$|t-\tau|=O(|d(t)||y_{112}(t)|)$
 for sufficiently small $d(t)$ and bounded $|y_{112}|$.
 
 For large $R_9>0$, consider the set $\{ t\in\mathbf{C} \mid |y_{112}|\le R_9\}$.
 Its connected component containing $\tau$ is an approximate disk $D_9$ with centre $\tau$ and radius $|d(\tau)|R_9$,
 and $t\mapsto y_{112}(t)$ is a complex analytic diffeomorphism from that approximate disk onto $\{y\in\mathbf{C}\mid|y|\le R_9\}$.
 \end{lemma}
 
 \begin{proof}
For the study of the solutions near $\mathcal{L}_9^p\setminus\mathcal{L}_8^p$, we use coordinates $(y_{112},z_{112})$, see Section \ref{a10-blow}.
In this chart, the line $\mathcal{L}_9^p\setminus\mathcal{L}_8^p$ is given by the equation $z_{112}=0$ and parametrised by $y_{112}\in\mathbf{C}$.
Moreover, $\mathcal{L}_{10}^p$ is given by $y_{112}=0$ and parametrised by $z_{112}\in\mathbf{C}$.
 
Asymptotically, for $z_{112}\to0$ and bounded $y_{112}$, $e^{-t}$, we have:
\begin{subequations}
\begin{align}
y_{112}' &\sim\frac{1}{z_{112}},\label{eq:y112'}
\\
z_{112}' &\sim-3(1+\eta)z_{112}-\frac{2}{\theta_1 e^t}y_{112}z_{112},\label{eq:z112'}
\\
J_{112} &\sim-\theta_1 e^t z_{112},\label{eq:J112}
\\
\frac{J'_{112}}{J_{112}} & = -2-\eta+\theta_1e^t+O(z_{112})=-2-\eta+\theta_1e^t+O(J_{112}),\label{eq:J'/J112}
\\
\frac{J_{112}}{J_{102}} & = 1+\frac{(1+\eta)\theta_1 e^t}{y_{112}}.\label{eq:J112/J102}
\end{align}
\end{subequations}

Integrating (\ref{eq:J'/J112}) from $\tau$ to $t$, we get
\begin{gather*}
J_{112}(t)=J_{112}(\tau)e^{-(2+\eta)(t-\tau)}e^{\theta_1(e^t-e^{\tau})}(1+o(1)).
\end{gather*}

Because of (\ref{eq:J112}), $z_{112}$ is approximately equal to a small constant, and from (\ref{eq:y112'}) follows that:
$$
y_{112}\sim y_{112}(\tau)+\frac{t-\tau}{z_{112}}.
$$
Thus, if $t$ runs over an approximate disk $D$ centered at $\tau$ with radius $|z_{112}|R$, then $y_{112}$ fills and approximate disk centered at $y_{112}(\tau)$ with radius $R$.
Therefore, if $z_{112}(\tau)\ll\tau$, the solution has the following properties for $t\in D$:
$$
\frac{z_{112}(t)}{z_{112}(\tau)}\sim1,
 $$ 
 and $y_{112}$ is a complex analytic diffeomorphism from $D$ onto an approximate disk with centre $y_{112}(\tau)$ and radius $R$.
 If $R$ is sufficiently large, we will have $0\in y_{112}(D)$, i.e.~the solution of the Painlev\'e equation will have a pole at a unique point in $D$.
 
 Now, it is possible to take $\tau$ to be the pole point.
 For $|t-\tau|\ll|\tau|$, we have:
 \begin{gather*}
 \frac{d(t)}{d(\tau)}\sim1,
 \quad\text{i.e.}\quad
 \frac{z_{112}(t)}{d(\tau)}\sim-\frac{1}{\theta_1 e^t}\cdot\frac{J_{112}(t)}{d(\tau)}\sim\frac{1}{\theta_1 e^t},
 \\
 y_{112}(t)\sim\frac{t-\tau}{z_{112}}\sim\frac{t-\tau}{d(\tau)}\cdot \theta_1e^t.
 \end{gather*}
 Let $R_9$ be a large positive real number.
 Then the equation $|y_{112}(t)|=R_9$ corresponds to $|(t-\tau)\theta_1 e^t|\sim|d(\tau)|R_3$, which is still small compared to $|\tau|$ if $|d(\tau)|$ is sufficiently small.
 Denote by $D_9$ the connected component of the set of all $t\in\mathbf{C}$ such that $\{t\mid |y_{112}(t)|\le R_9\}$ is an approximate disk with centre $\tau$ and radius $2|d(\tau)|R_3$.
More precisely, $y_{112}$ is a complex analytic diffeomorphism from $D_9$ onto $\{y\in\mathbf{C}\mid|y|\le R_9\}$, and
$$
\frac{d(t)}{d(\tau)}\sim1
\quad\text{for all}\quad
t\in D_9.
$$
From (\ref{eq:J112/J102}), we have:
$$
\frac{J_{112}}{J_{102}} \sim 1
\quad
\text{when}
\quad
1
\gg
\left|\frac{(1+\eta)\theta_1 e^t}{y_{112}(t)}\right|
\sim
\left| \frac{(1+\eta) d(\tau)}{t-\tau}   \right|,
$$
that is, when
$$
|t-\tau|\gg|d(\tau)|.
$$
Since $R_9\gg1$, we have
$$
|t-\tau|\sim|d(\tau)|R_9\gg|d(\tau)|.
$$
Thus $\dfrac{J_{112}}{J_{102}}\sim1$ for the annular disk $t\in D_9\setminus D_9'$, where $D_9'$ is a disk centered at $\tau$ with small radius compared to the radius of $D_9$.
\end{proof}

\begin{lemma}[Behaviour near $\mathcal{L}_8^p\setminus\mathcal{L}_1^p$]\label{lemma:L8}
For large finite $R_8>0$, consider the set of all $t\in\mathbf{C}$, such that the solution at complex time $t$ is close to 
$\mathcal{L}_8^p\setminus\mathcal{L}_1^p$, with $|y_{102}(t)|\le R_8$, but not close to $\mathcal{L}_9^p$.
Then this set is the complement of $D_9$ in an approximate disk $D_8$ with centre $\tau$ and radius $\sim\sqrt{d(\tau)}R_8$.
More precisely, $t\mapsto y_{102}$ defines a covering from the annular domain $D_8\setminus D_9$ onto the complement in
$\{ z\in\mathbf{C}\mid |z|\le R_8\}$ of an approximate disk with centre at the origin and small radius $\sim |d(\tau)|R_9$, where
$y_{102}(t)\sim d(\tau)^{-1/2}(t-\tau)$.
 \end{lemma}
 
 \begin{proof}
 Set $\mathcal{L}_8^p\setminus\mathcal{L}_1^p$ is visible in the chart $(y_{102},z_{102})$, where it is given by the equation $z_{102}=0$ and parametrized by $y_{102}\in\mathbf{C}$, see Section \ref{a9-blow}.
 In that chart, the line $\mathcal{L}_9^p$ (without one point) is given by the equation $y_{102}=0$ and parametrized by $z_{102}\in\mathbf{C}$.
 The line $\mathcal{L}_4^p$ (without one point) is given by the equation $y_{102}=-\theta_1 e^t$ and also parametrized by $z_{102}\in\mathbf{C}$.
 
For $z_{102}\to0$, bounded $e^t$, and $y_{102}$ bounded and bounded away from $-\theta_1 e^t$, we have:
\begin{subequations}
\begin{align}
z_{102}' &\sim-\frac{1}{y_{102}}-\frac{1}{\theta_1e^t+y_{102}}-\theta_1e^t-2y_{102},\label{eq:z102'}
\\
y_{102}' &\sim\frac{2}{z_{102}},\label{eq:y102'}
\\
J_{102} &=-y_{102}(\theta_1 e^t +y_{102}) z_{102}^2,\label{eq:J102}
\\
\frac{J'_{102}}{J_{102}} & \sim \frac{(1+\eta)\theta_1 e^t}{y_{102}}+\theta_1 e^t+2y_{102},\label{eq:J'/J102}
\\
E J_{102} & \sim 1-\frac{\theta_1 e^t}{y_{102}}.\label{eq:EJ102}
\end{align}
\end{subequations}


From (\ref{eq:J'/J102}):
\begin{gather*}
\log\frac{J_{102}(t_1)}{J_{102}(t_0)}
\sim
\theta_1(e^{t_1}-e^{t_0})+(t_1-t_0)K,
\\
K=\frac{(1+\eta)\theta_1 e^{\tilde\tau}}{y_{102}(\tilde\tau)}+2y_{102}(\tilde\tau),
\end{gather*}
where $\tilde\tau$ is on the integration path.

Therefore $J_{102}(t_1)/J_{102}(t_0)\sim1$, if for all $t$ on the segment from $t_0$ to $t_1$ we have $|t-t_0|\ll|t_0|$ and
$$
\left|\frac{\theta_1 e^t}{y_{102}(t)}\right|\ll\frac1{|t_0|},
\qquad
|y_{102}(t)|\ll\frac1{|t_0|}.
$$
We choose $t_0$ on the boundary of $D_9$ from Lemma \ref{lemma:L9}.
Then we have
$$
\frac{d(\tau)}{d(t_0)}\sim \frac{J_{112}(\tau)}{J_{102}(t_0)}\sim1
\quad\text{and}\quad
|y_{112}(t_0)|=R_9,
$$
which implies that
$$
|z_{102}|=\left|\frac{1}{y_{112}+(1+\eta)\theta_1 e^t}\right|\sim\frac1{R_9}\ll 1.
$$

Since $D_9$ is an approximate disk with centre $\tau$ and small radius $\sim|d(\tau)|R_9$, and $R_9\gg|\tau|^{-1}$, we have that
$|y_{112}(t)|\ge R_9\gg 1$ hence:
$$
|y_{102}|\ll1
\quad\text{if}\quad
t=\tau+r(t_0-\tau),
\ r\ge1,
$$
and
$$
\frac{|t-t_0|}{|t_0|}=(r-1)\left|1-\frac{\tau}{t_0}\right|\ll1
\quad\text{if}\quad
r-1\ll\frac{1}{|1-\frac{\tau}{t_0}|}.
$$

Then equation (\ref{eq:J102}) and $J_{102}\sim-d(\tau)$ yield
$$
z_{102}^{-1}
\sim
\left(\frac{y_{102}(\theta_1e^t+y_{102})}{d(\tau)}\right)^{1/2}
\sim
\left(\frac{y_{102}^2}{d(\tau)}\right)^{1/2},
$$
which in combination with (\ref{eq:z102'}) leads to
$$
\frac{d y_{102}}{dt}\sim d(\tau)^{-1/2}.
$$
Hence
$$
y_{102}(t)\sim y_{102}(t_0)+d(\tau)^{-1/2}(t-t_0),
$$
and therefore
$$
y_{102}(t) \sim d(\tau)^{-1/2}(t-t_0)
\quad\text{if}\quad
|t-t_0|\gg|y_{102}(t_0)|.
$$
For large finite $R_8>0$, the equation $|y_{102}|=R_8$ corresponds to $|t-t_0|\sim\sqrt{d(\tau)}R_8$, which is still small compared to 
$|t_0|\sim|\tau|$, therefore $|t-\tau|\le|t-t_0|+|t_0-\tau|\le|\tau|$.
This proves the statement of the lemma.
\end{proof}

\begin{lemma}[Behaviour near $\mathcal{L}_4^p$]\label{lemma:L4}
For large finite $R_4>0$, consider the set of all $t\in\mathbf{C}$, such that the solution at complex time $t$ is close to 
$\mathcal{L}_4^p$, with $|z_{91}(t)|\le R_4$, but not close to $\mathcal{L}_8^p$.
Then this set is the complement of $D_8$ in an approximate disk $D_4$ with centre $\tau$ and radius $\sim R_4/d(\tau)$.
More precisely, $t\mapsto z_{91}$ defines a covering from the annular domain $D_4\setminus D_8$ onto the complement in
$\{ z\in\mathbf{C}\mid |z|\le R_4\}$ of an approximate disk with centre at the origin and small radius $\sim |d(\tau)|/R_8$, where
$z_{91}(t)\sim d(\tau)/(t-\tau)$.
\end{lemma}
\begin{proof}
$\mathcal{L}_4^p$ without one point is visible in the chart $(y_{91},z_{91})$, where it is given by the equation $y_{91}=0$ and parametrized by $z_{91}\in\mathbf{C}$, see Section \ref{a8-blow}.
In that chart, the line $\mathcal{L}_1$ is not visible, while $\mathcal{L}_8$ is given by the equation $z_{91}=0$.

For $y_{91}\to0$ and bounded $z_{91}$ and $e^t$, we have:
\begin{subequations}
\begin{align}
y_{91}' &\sim-\frac{2\theta_1 e^t}{z_{91}},\label{eq:y91'}
\\
z_{91}' &\sim\frac{\theta_1 e^t}{y_{91}},\label{eq:z91'}
\\
J_{102}& = \frac{y_{91}z_{91}^2}{\theta_1 e^t-y_{91}},\label{eq:J91}
\\
\frac{J_{102}'}{J_{102}} & \sim -1-\eta -\theta_1 e^t
.\label{eq:J'/J91}
\end{align}
\end{subequations}
From (\ref{eq:J'/J91}):
$$
\log\frac{J_{102}(t_1)}{J_{102}(t_0)}\sim-(1+\eta)(t_1-t_0)-\theta_1(e^{t_1}-e^{t_0}).
$$
Therefore $J_{102}(t_1)/J_{102}(t_0)\sim1$ if for all $t$ on the segment from $t_0$ to $t_1$ we have $|t-t_0|\ll t_0$.
We choose $t_0$ on the boundary of $D_8$ from Lemma \ref{lemma:L8}.
Then we have
$$
\frac{d(\tau)}{d(t_0)}\sim\frac{J_{102}(\tau)}{J_{102}(t_0)}\sim1
\quad\text{and}\quad
|y_{102}(t_0)|=R_8,
$$
which implies that
$$
|\theta_1e^t|=|y_{91}-y_{102}|\sim R_8\gg1.
$$
Hence:
$$
|y_{91}|\ll 1
\quad\text{if}\quad
|\theta_1e^t|\sim R_8.
$$
Then equation (\ref{eq:J'/J91}) and $J_{102}\sim-d(\tau)$ yield
$$
y_{91}^{-1}\sim-\frac{z_{91}^2}{\theta_1e^t d(\tau)},
$$
then, using (\ref{eq:z91'}) we get:
$$
\frac{d(z_{91}^{-1})}{dt}\sim\frac1{d(\tau)}.
$$
It follows that
$$
z_{91}^{-1}\sim z_{91}(t_0)^{-1}+\frac{t-t_0}{d(\tau)},
$$
and therefore
$$
z_{91}\sim\frac{d(\tau)}{t-t_0}
\quad\text{if}\quad
|t-t_0|\gg|z_{91}(t_0)^{-1}|.
$$
For large finite $R_4>0$, the equation $|z_{91}|=R_4$ corresponds to $|t-t_0|\sim d(\tau)/R_4$, which is small compared to 
$|t_0|\sim|\tau|$, and therefore $|t-\tau|\le|t-t_0|+|t_0-\tau|\le|\tau|$.
This proves the statement of the lemma.
\end{proof}

\begin{theorem}\label{th:estimates}
Let $\epsilon_1$, $\epsilon_2$, $\epsilon_3$ be given such that $\epsilon_1>0$, $0<\epsilon_2<|\theta_1|$, $0<\epsilon_3<1$.
Then there exists $\delta_1>0$ such that if $|e^{t_0}|<\epsilon_1$ and $|d(t_0)|<\delta_1$, then:
$$
\rho=\inf\{ r<|e^{t_0}| \ \text{such that}\ |d(t)|<\delta_1\ \text{whenever}\ |e^{t_0}|\ge|e^{t}|\ge r\}
$$
satisfies:
\begin{itemize}
\item[(i)]
$\delta_1\ge|d(t_0)|\left(\rho^{-1}|e^{t_0}|\right)^{\theta_1-\epsilon_2}(1-\epsilon_3)$;
\item[(ii)]
if $|e^{t_0}|\ge|e^{t}|\ge\rho$ then $d(t)=d(t_0)e^{\theta_1(t-t_0)+\epsilon_2(t)}(1+\epsilon_3(t))$;
\item[(iii)]
if $|e^{t}|\le\rho$ then $d(t)\ge\delta_1(1-\epsilon_3)$.
\end{itemize}
\end{theorem}

\begin{proof}
Suppose a solution of the system (\ref{eq:PVlog-system}) is close to the infinity set at times $t_0$ and $t_1$.
It follows from Lemmas \ref{lemma:L3}--\ref{lemma:L4} that for every solution close to $\mathcal{I}$, the set of complex time $t$ such that the solution is not close to $\mathcal{L}_{1}^p\cup\mathcal{L}_{2}^p$ is the union of approximate disks of radius $\sim|d|$.
Hence if the solution is near $\mathcal{I}$ for all complex times $t$ such that $|e^{t_0}|\ge|e^{t}|\ge|e^{t_1}|$, then there exists a path $\mathcal{P}$ from $t_0$ to $t_1$, such that the solution is close to $\mathcal{L}_{1}^p\cup\mathcal{L}_{2}^p$ for all $t\in\mathcal{P}$ and $\mathcal{P}$ is $C^1$-close to the path: $s\mapsto t_1^s t_0^{1-s}$, $s\in[0,1]$.

Then Lemmas \ref{lemma:L2} and \ref{lemma:L1} imply that:
$$
\log\frac{E(t)}{E(t_0)}=-\theta_1(t-t_0)\int_0^1dt+o(1).
$$
Therefore
$$
E(t)=E(t_0)e^{-\theta_1(t-t_0)+o(1)}(1+o(1)),
$$
and
\begin{equation}\label{eq:d(t)d(t0)}
d(t)=d(t_0)e^{\theta_1(t-t_0)+o(1)}(1+o(1)).
\end{equation}
From Lemmas \ref{lemma:L3}--\ref{lemma:L4}, we then have that, as long as the solution is close to $\mathcal{I}$, the ratio of $d$ remains close to $1$.

For the first statement of the theorem, we have:
$$
\delta_1>|d(t)|\ge |d(t_0)| \left| e^{\theta_1(t-t_0)-\epsilon_2}\right|(1-\epsilon_3)
$$
and so
$$
\delta_1\ge\sup_{\{t\mid |d(t)|<\delta_1\}}  |d(t_0)| \left| e^{\theta_1(t-t_0)-\epsilon_2}\right|(1-\epsilon_3).
$$
The second statement follows from (\ref{eq:d(t)d(t0)}), while the third follows by the assumption on $t$.
\end{proof}

In the following corollary, we summarise the results obtained in this section.

\begin{corollary}\label{cor:infinity}
No solution of $(\ref{eq:PVlog-system})$ intersects $\mathcal{I}$.
A solution that approached $\mathcal{I}$ will stay in its vicinity for a limited range of the independent variable $t$.
Moreover, if a solution is sufficiently close to $\mathcal{I}$ at a point $t$, then it will have a pole in a neighbourhood of $t$. 
\end{corollary}

\begin{proof}
The first two statements follow from Theorem \ref{th:estimates}, and the last one from Lemmas \ref{lemma:L2}--\ref{lemma:L4}.
\end{proof}

\section{The limit set}
\label{sec:limit}
In this section we consider properties of the limit set of the solutions, when $\Real t\to-\infty$, i.e.~$x\to0$.
First, we define the limit set, generalising the concept of limit sets in dynamical systems.

\begin{definition}
Let $(y(t),z(t))$ be a solution of (\ref{eq:PVlog-system}). \emph{The limit set} $\Omega_{(y,z)}$ of $(y(t),z(t))$ is the set of all $S\in\mathcal{F}_{\infty}\setminus\mathcal{I}_{\infty}$ such that there exists a sequence $t_n\in\mathbf{C}$ satisfying:
$$
\lim_{n\to\infty}\Real t_n=-\infty
\quad\text{and}\quad
\lim_{n\to\infty}(y(t_n),z(t_n))=S.
$$
\end{definition}

\begin{theorem}\label{th:limit}
There exists a compact subset $K$ of $\mathcal{F}_{\infty}\setminus\mathcal{I}_{\infty}$, such that the limit set $\Omega_{(y,z)}$ of any solution $(y,z)$ is contained in $K$.
Moreover, $\Omega_{(y,z)}$ is a non-empty, compact and connected set, which is invariant under the flow of the autonomous system (\ref{eq:PVlog-system-auto}).
\end{theorem}

\begin{proof}
For any positive numbers $\delta_1$, $r$, let $K_{\delta_1,r}$ denote the set of all $s\in\mathcal{F}(t)$ such that $|e^t|\le r$ and $d(s)\ge\delta_1$.
Since $\mathcal{F}(t)$ is a complex analytic family over $\mathbf{C}$ of compact surfaces, $K_{\delta_1,r}$ is also compact.
Furthermore $K_{\delta_1,r}$ is disjoint from the union of the infinity sets $\mathcal{I}(t)$, $t\in\mathbf{C}$, and therefore $K_{\delta_1,r}$ is a compact susbset of the Okamoto space $\mathcal{O}\setminus\mathcal{F}_{0}$.
When $r$ approches zero, the sets $K_{\delta_1,r}$ shrink to the set
$$
K_{\delta_1,0}=\{ s\in\mathcal{F}(0) \mid |d(s)|\ge\delta_1\}\subset\mathcal{F}_{\infty}\setminus\mathcal{I}_{\infty},
$$
which is compact.

It follows from Theorem \ref{th:estimates} that there exists $\delta_1\ge0$ such that for every solution $(y,z)$ there exists $r_0>0$ with the following property:
$$
(y(t),z(t))\in K_{\delta,r_0}
\quad
\text{for every}\ t\ \text{such that}\ |e^t|\le r_0.
$$
In the sequel, we take $r\le r_0$, when it follows that $(y(t),z(t))\in K_{\delta_1,r}$ whenever $|e^t|\le r$.
Let 
$T_r=\{t\in\mathbf{C}\mid |e^t|\le r\}$
and let $\Omega_{(y,z),r}$ denote the closure of $(y,z)(T_r)$ in $\mathcal{O}$.
Since $T_r$ is connected and $(y,z)$ continuous, $\Omega_{(y,z),r}$ is also connected.
Since $(y,z)(T_r)$ is contained in the compact subset $K_{\delta_1,r}$, its closure $\Omega_{(y,z),r}$ is also contained there and therefore $\Omega_{(y,z),r}$ is a non-empty compact and connected subset of $\mathcal{O}\setminus\mathcal{F}(0)$.
The intersection of a decreasing sequence of non-empty compact and connected sets is non-empty, compact and connected: therefore, as $\Omega_{(y,z),r}$ decrease to $\Omega_{(y,z)}$ when $r$ tends to zero, it follows that $\Omega_{(y,z)}$ is a non-empty, compact and connected set of $\mathcal{O}$.
Since $\Omega_{(y,z),r}\subset K_{\delta_1,r}$ for all $r\le r_0$, and the sets $K_{\delta_1,r}$ shrink to the compact subset $K_{\delta_1},0$ of $\mathcal{F}_{\infty}\setminus\mathcal{I}_{\infty}$ as $r$ tends to zero, it follows that $\Omega_{(y,z)}\subset K_{\delta_1,0}$.
This proves the first statement of the theorem with $K=K_{\delta_1,0}$.

Since $\Omega_{(y,z)}$ is the intersection of the decreasing family of compact sets $\Omega_{(y,z),r}$, there exists for every neighbourhood $A$ of $\Omega_{(y,z)}$ in $\mathcal{O}$ and $r>0$ such that $\Omega_{(y,z),r}\subset A$, hence $(y(t),z(t))\in A$ for every $t\in\mathbf{C}$ such that $|e^t|\le r$.
If $t_j$ is any sequence in $\mathbf{C}$ such that $\Real t_j\to-\infty$, then the compactness of $K_{\delta_1,r}$, in combination with $(y,z)T_r\subset K_{\delta_1,r}$, implies that there is a subsequence $j=j(k)\to\infty$ as $k\to\infty$, such that:
$$
(y(t_{j(k)}),z(t_{j(k)}))\to s
\ \ \text{as}\ \ k\to\infty.
$$
Then it follows that $s\in\Omega_{(y,z)}$.

Next we prove that $\Omega_{(y,z)}$ is invariant under the flow $\Phi^{\tau}$ of the autonomous Hamiltonian system (\ref{eq:PVlog-system-auto}).
Let $s\in\Omega_{(y,z)}$ and $t_j$ be a sequence in $\mathbf{C}$ such that $\Real t_j\to-\infty$ and $(y(t_j),z(t_j))\to s$.
Since the $t$-dependent vector field of the system (\ref{eq:PVlog-system}) converges in $C^1$ to the vector field of the autonomous system (\ref{eq:PVlog-system-auto}) as $\Real t\to-\infty$, it follows from the continuous dependence on initial data and parameters, that the distance between $(y(t_j+\tau),z(t_j+\tau))$ and $\Phi^{\tau}(y(t_j),z(t_j))$ converges to zero as $j\to\infty$.
Since $\Phi^{\tau}(y(t_j),z(t_j))\to\Phi^{\tau}(s)$ and $\Real t_j\to-\infty$ as $j\to\infty$, 
it follows that $(y(t_j+\tau),z(t_j+\tau))\to\Phi^{\tau}(s)$ and $t_j+t\to\infty$ as $j\to\infty$, hence $\Phi^{\tau}(s)\in\Omega_{(y,z)}$.
\end{proof}

\begin{proposition}\label{prop:intersections}
If $y$ is a solution of (\ref{eq:PV}) with essential singularity at $x=0$, than the flow $(y,z)$ of the vectory field (\ref{eq:PVlog-system})  meets each of the pole lines $\mathcal{L}_5$, $\mathcal{L}_6$, $\mathcal{L}_7$ infinitely many times.
\end{proposition}

\begin{proof}
First, suppose that a solution $(y(t),z(t))$ intersects the union $\mathcal{L}_5\cup\mathcal{L}_6\cup\mathcal{L}_7$ only finitely many times.

According to Theorem \ref{th:limit}, the limit set $\Omega_{(y,z)}$ is a compact set in $\mathcal{F}_{\infty}\setminus\mathcal{I}_{\infty}$.
If $\Omega_{(y,z)}$ intersects one of the pole lines $\mathcal{L}_5$, $\mathcal{L}_6$, $\mathcal{L}_7$ at a point $p$, then there exists $t$ with arbitrarily large negative real part such that $(u(t),v(t))$ is arbitrarily close to $p$, when the transversality of the vector field to the pole line implies thet $(y(\tau),z(\tau))\in\mathcal{L}_5\cup\mathcal{L}_6\cup\mathcal{L}_7$ for a unique $\tau$ near $t$.
As this would imply that $(y(t),z(t))$ intersects $\mathcal{L}_5\cup\mathcal{L}_6\cup\mathcal{L}_7$ infinitely many times, it follows that $\Omega_{(y,z)}$ is a compact subset of $\mathcal{F}_{\infty}\setminus(\mathcal{I}_{\infty}\cup\mathcal{L}_5\cup\mathcal{L}_6\cup\mathcal{L}_7)$.
It follows that $\Omega_{(y,z)}$ is a compact subset contained in the first affine chart, which implies that $y$ and $z$ remain bounded for large negative $\Real t$.
Thus $y$, $z$ are holomorphic functions of $x=e^t$ in a neighbourhood of $x=0$, which implies that there are complex numbers $y(\infty)$, $z(\infty)$ which are the limit points of $y(t)$, $z(t)$ as $\Real t\to-\infty$.
That means that $y$ is analytic at $x=0$, which contradicts the assumption that it has there an essential singularity.

Since the limit set $\Omega_{(y,z)}$ is invariant under the autonomous flow, it means that it will contain the whole irreducible component of a curve from the pencil $h_c(y,z)=0$ given by (\ref{eq:pencil}), for some constant $c$.
It is shown in Section \ref{sec:conics} that this pencil of curves is birationally equivalent to a pencil of conics.
We identified in Section \ref{sec:conics} the three singular conics in the pencil and found the special solutions corresponding to them.

In all other cases, all three base points $b_5$, $b_6$, $b_7$ will be contained in the limit set, which are projections of the pole lines $\mathcal{L}_5(\infty)$, $\mathcal{L}_6(\infty)$, $\mathcal{L}_7(\infty)$ respectively.
For a general solution $(y,z)$, the base point $b_4$ will not be contained in the limit set, because that point is not a base point of the autonomous system (\ref{eq:PVlog-auto}).
\end{proof}

\begin{remark}
If the limit set $\Omega_{(y,z)}$ contains only one point, that point must be a fixed point of the autonomous system (\ref{eq:PVlog-system-auto}).
As we obtained in Section \ref{sec:auto}, there are four such points.
One of the points has $y$-coordinate equal to unity and it corresponds to the rational solutions of the form $\dfrac{\kappa}{x+\kappa}$ and $\dfrac{\kappa+x}{\kappa-x}$.
\end{remark}

\begin{theorem}\label{th:zeroespoles}
Every solution of (\ref{eq:PV}) with essential singularity at $x=0$ has infinitely many poles and infinitely many zeroes in each neighbourhood of that singular point.
\end{theorem}

\begin{proof}
Applying results from Section \ref{sec:poles} and Proposition \ref{prop:intersections}, we get that each solution has a simple pole at the intersections with $\mathcal{L}_5$ and $\mathcal{L}_6$ and a simple zero at the intersection with $\mathcal{L}_7$.
\end{proof}

\section{Limit $x\to\infty$}
\label{sec:x_infinity}
For studying the limit $x\to\infty$, it is convenient to represent the $\PV$ as the following system:
\begin{equation}\label{eq:PVsystem}
\begin{split}
&y'=\frac1x\big(
2y(y-1)^2z-(\theta_0+\eta)y^2+(2\theta_0+\eta-\theta_1x)y-\theta_0
    \big),
\\
&z'=-\frac1x\bigg(
(y-1)(3y-1)z^2-\big(2(\theta_0+\eta)y-2\theta_0-\eta+\theta_1x\big)z
+\frac12\epsilon(\theta_0+\eta-\theta_{\infty})
\bigg),
\end{split}
\end{equation}
where
$\theta_{\infty}^2=2\alpha$,
$\theta_0^2=-2\beta$,
$\theta_1^2=-2\delta$ $(\theta_1\neq0)$,
$\eta=-\frac{\gamma}{\theta_1}-1$,
$\epsilon=\frac12(\theta_0+\theta_{\infty}+\eta)$.

\begin{remark}
Using the change of the independent variable $t=\log x$, equation (\ref{eq:PVsystem}) will give $\PV$ in the form (\ref{eq:PVlog}). 
\end{remark}

The resolution of the singularities of (\ref{eq:PVsystem}) will lead to the same space of the initial values as described in Section \ref{sec:space}, and shown in Figure \ref{fig:okamoto-blow-down}.

In the limit $x\to\infty$, the fifth Painlev\'e equation (\ref{eq:PV}) becomes:
\begin{equation}\label{eq:PV-auto}
y''=\left(\frac{1}{2y}+\frac{1}{y-1}\right) y'^2+\frac{\delta y(y+1)}{y-1},
\end{equation}
which has a first integral:
$$
\frac{y'^2}{2y(y-1)}+\frac{\delta y}{(y-1)^2}.
$$
The solutions of (\ref{eq:PV-auto}) are elliptic functions satisfying:
$$
y'^2=2y\left(C(y-1)^2 - \delta y\right),
$$
where $C$ is an arbitrary constant.
Notice that $y\equiv0$ is the only solution of that equation taking the value $0$.
Thus, in contrast to the case when $x\to0$, $a_3$ is not a base point for the autonomous equation, while $a_4$ will be a base point for that equation.

Now, analysing the system (\ref{eq:PVsystem}) in the similar way as shown in Sections \ref{sec:infinity} and \ref{sec:limit}, 
it follows in the limit $x\to\infty$, that a general solution of that system has a compact limit set which is invariant with the respect to the autonomous flow.
This implies, as in the proof of Proposition \ref{prop:intersections}, that the flow $(y,z)$ of the vector field (\ref{eq:PVsystem}) meets each of the pole lines $\mathcal{L}_5$, $\mathcal{L}_6$, $\mathcal{L}_{10}$ infinitely many times.
Thus every solution of (\ref{eq:PV}) with essential singularity at $x=\infty$ has infinitely many poles and take the value $1$ infinitely many times in each neighbourhood of that singular point.


\appendix

\section{Resolution of the system}
\label{sec:resolution}
In this section, we give the explicit construction of the space of initial conditions for (\ref{eq:PVlog-system}).
This constructions consists of eleven successive blow-ups of points in $\mathbf{CP}^2$.

We use the following notation.
The coordinates in three affine charts of $\mathbf{CP}^2$ are denoted by $(y_{01},z_{01})$,  $(y_{02},z_{02})$, and  $(y_{03},z_{03})$.
The exceptional line obtained in the $n$-th blow-up is covered by two coordinate charts, denoted by $(y_{n1},z_{n1})$ and $(y_{n2},z_{n2})$.

In each of these charts, we write the system (\ref{eq:PVlog-system}) in the corresponding coordinates and look for \emph{base points} -- the points contained by infinitely many solutions.
We calculate the coordinates of the base points in local coordinates in the following way.
In each chart $(y_{nj},z_{nj})$, the system can be written in the form:
$$
y_{nj}'=\frac{P(y_{nj},z_{nj},e^t)}{Q(y_{nj},z_{nj},e^t)},
\quad
z_{nj}'=\frac{R(y_{nj},z_{nj},e^t)}{S(y_{nj},z_{nj},e^t)},
$$
for some polynomial expressions $P$, $Q$, $R$, $S$.
The uniqueness property for the given initial conditions is broken whenever $P=Q=0$ or $R=S=0$,
so solving these equations yields to base points.
We note that, after blowing up, a new base points can appear only on the exceptional line.

We remark that a base point in algebraic geometry is a joint point of all curves from a given pencil.
The solutions of the autonomous system (\ref{eq:PVlog-system-auto}) are algebraic curves from the pencil (\ref{eq:pencil}), hence the notions of base points of a system of differential equations and base points of a pencil of curves coincide in the autonomous case.


\subsection{Affine chart $(y_{01},z_{01})$}\label{chart01}
The first affine chart is defined by the original coordinates:
$y_{01}=y$,
$z_{01}=z$.
The energy (\ref{eq:E}) is:
$$
 E=y(y-1)^2z^2-(\theta_0+\eta)y^2z+(2\theta_0+\eta)yz-\theta_0z+\frac12\epsilon(\theta_0+\eta-\theta_{\infty})y.
$$

\subsection{Affine chart $(y_{02},z_{02})$}\label{chart02}

The second affine chart is given by the coordinates:
\begin{gather*}
y_{02}=\frac1{y},\quad  z_{02}=\frac{z}{y},
\\
y=\frac1{y_{02}},\quad z=\frac{z_{02}}{y_{02}}.
\end{gather*}
The line at infinity is $\mathcal{L}_{\infty} : y_{02}=0$.

The Painlev\'e vector field (\ref{eq:PVlog-system}) is:
$$
\begin{aligned}
y_{02}'&= - \frac{ 2 z_{02}}{y_{02}^2} + \frac{4 z_{02}}{y_{02}}+\eta + \theta_0 - (\eta + 2 \theta_0) y_{02} + \theta_0 y_{02}^2 - 2 z_{02} 
+\theta_1 e^t y_{02},
\\
z_{02}'&= - \frac{5 z_{02}^2}{y_{02}^3} + \frac{8 z_{02}^2}{y_{02}^2} + 
 y_{02} ( \theta_0 z_{02}-F) + \frac{3z_{02} (\eta + \theta_0- z_{02})}{y_{02}}-2 (\eta+2\theta_0) z_{02}
 + 2 \theta_1 e^t  z_{02}.
\end{aligned}
$$
In this chart, there is one visible base point:
$$
a_0(y_{02}=0,z_{02}=0).
$$

\subsection{Affine chart $(y_{03},z_{03})$}\label{chart03}

The coordinates:
\begin{gather*}
y_{03}=\frac{y}{z},
\quad
z_{03}=\frac{1}{z},
\\
y=\frac{y_{03}}{z_{03}},
\quad
z=\frac1{z_{03}}.
\end{gather*}
The line at infinity is $\mathcal{L}_{\infty} : z_{03}=0$.

The vector field (\ref{eq:PVlog-system}) is:
$$
\begin{aligned}
y_{03}'&=  \frac{5 y_{03}^3}{z_{03}^3} - \frac{8 y_{03}^2}{z_{03}^2} + 
\frac{ 3 y_{03}(1-( \eta+\theta_0) y_{03})}{z_{03}} + ( F y_{03}-\theta_0 ) z_{03}
+2 (\eta+2\theta_0) y_{03}
-2 \theta_1  e^t   y_{03}
,
\\
z_{03}'&=\frac{3 y_{03}^2}{z_{03}^2} - \frac{4 y_{03}}{z_{03}} + 
 (\eta + 2 \theta_0) z_{03} + F z_{03}^2 +1 - 2 (\eta + \theta_0) y_{03} 
 - \theta_1 e^t z_{03}.
\end{aligned}
$$
In this chart, there is one base point:
$$
a_1(y_{03}=0,z_{03}=0).
$$


\subsection{Resolution at $a_0$}\label{a0-blow}
\paragraph{First chart:}
\begin{gather*}
y_{11}=\frac{y_{02}}{z_{02}}=\frac{1}{z},
\quad
z_{11}=z_{02}=\frac{z}{y},
\\
y=\frac{1}{y_{11}z_{11}},
\quad
z=\frac{1}{y_{11}}.
\end{gather*}
The exceptional line is $\mathcal{L}_0 : z_{11}=0$, while the proper preimage of the line at infinity $\mathcal{L}_{\infty}$ is given by $y_{11}=0$. 

The vector field (\ref{eq:PVlog-system}) in this chart is:
$$
\begin{aligned}
y_{11}'&=
\frac{3}{y_{11}^2 z_{11}^2} - \frac{ 2 (\eta+ \theta_0)}{z_{11}} - \frac{4}{y_{11} z_{11}}+
1 + (\eta  + 2 \theta_0) y_{11} + F y_{11}^2 - \theta_1 e^t y_{11}
,
\\
z_{11}'&=
 - \frac{5}{y_{11}^3 z_{11}}+
\frac{8}{y_{11}^2} + \frac{3( \eta+\theta_0)}{y_{11}}- \frac{3 z_{11}}{y_{11}} - 2( \eta+2 \theta_0) z_{11}  - F y_{11} z_{11} + \theta_0 y_{11} z_{11}^2+2  \theta_1 e^t z_{11}
.
\end{aligned}
$$
In this chart, there are no base points on the exceptional line $\mathcal{L}_0$.

\paragraph{Second chart:}
\begin{gather*}
y_{12}=y_{02}=\frac{1}{y},
\quad
z_{12}=\frac{z_{02}}{y_{02}}=z,
\\
y=\frac{1}{y_{12}},
\quad
z=z_{12}.
\end{gather*}
The exceptional line is $\mathcal{L}_0 : y_{12}=0$, while the proper preimage of the line at infinity $\mathcal{L}_{\infty}$ is not visible in this chart.

The vector field (\ref{eq:PVlog-system}) is:
$$
\begin{aligned}
y_{12}'&=
 -\frac{ 2 z_{12}}{y_{12}}+\eta + \theta_0 - (\eta + 2 \theta_0) y_{12} + \theta_0 y_{12}^2 + 4 z_{12}  - 2 y_{12} z_{12}
 +\theta_1 e^t y_{12}
,
\\
z_{12}'&=
- \frac{3 z_{12}^2}{y_{12}^2} + \frac{4 z_{12}^2}{y_{12}}+ \frac{2 (\eta + \theta_0) z_{12}}{y_{12}}-F - (\eta + 2 \theta_0) z_{12}  - z_{12}^2 
+\theta_1 e^t z_{12}
.
\end{aligned}
$$
In this chart, there is one base point on the exceptional line $\mathcal{L}_0$:
$$
a_2(y_{12}=0,z_{12}=0).
$$

\subsection{Resolution at $a_1$}\label{a1-blow}
\paragraph{First chart:}
\begin{gather*}
 y_{21}=\frac{y_{03}}{z_{03}}=y,
 \quad
z_{21}=z_{03}=\frac1z,
\\
y=y_{21},
\quad
z=\frac{1}{z_{21}}.
\end{gather*}
The exceptional line is $\mathcal{L}_1 : z_{21}=0$, while the proper preimage of the line at infinity $\mathcal{L}_{\infty}$ is not visible in this chart.

The vector field (\ref{eq:PVlog-system}) is:
$$
\begin{aligned}
y_{21}'&=
\frac{ 2 ( y_{21}-1)^2 y_{21}}{z_{21}}
-\theta_0 + (\eta + 2 \theta_0- \theta_1 e^t) y_{21} - (\eta + \theta_0) y_{21}^2 
,
\\
z_{21}'&=1 - 4 y_{21} + 3 y_{21}^2 + (\eta + 2 \theta_0-\theta_1e^t  ) z_{21} - 2 (\eta+ \theta_0) y_{21} z_{21} + F z_{21}^2
.
\end{aligned}
$$
In this chart, there are two base points on the exceptional line $\mathcal{L}_1$:
$$
a_3(y_{21}=0,z_{21}=0),
\quad
a_4(y_{21}=1,z_{21}=0).
$$

The energy (\ref{eq:E}) and its rate of change are:
$$
\begin{aligned}
E&=
\frac{y_{21} (y_{21}-1)^2}{z_{21}^2}
-\frac{(y_{21}-1) \big((\eta +\theta_0) y_{21}-\theta_0\big)}{z_{21}}
+F y_{21}
,
\\
E'&=
\theta_1 e^t
\left(
\frac{y_{21}-y_{21}^3}{z_{21}^2}
+
\frac{(\eta+\theta_0) y_{21}^2-\theta_0}{z_{21}}-F y_{21}
\right)
. 
\end{aligned}
$$

\paragraph{Second chart:}
\begin{gather*}
y_{22}=y_{03}=\frac{y}{z},
\quad
z_{22}=\frac{z_{03}}{y_{03}}=\frac1y,
\\
y=\frac1{z_{22}},
\quad
z=\frac1{y_{22}z_{22}}.
\end{gather*}
The exceptional line is $\mathcal{L}_1 : y_{22}=0$, while the proper preimage of the line at infinity $\mathcal{L}_{\infty}$ is $z_{22}=0$. 

The vector field (\ref{eq:PVlog-system}) is:
$$
\begin{aligned}
y_{22}'&=
\frac{5}{z_{22}^3} - \frac{8}{z_{22}^2} + \frac{3(1-( \eta+\theta_0) y_{22})}{z_{22}} 
- \theta_0 y_{22} z_{22} + F y_{22}^2 z_{22}+2 (\eta +2 \theta_0-\theta_1 e^t) y_{22} 
,
\\
z_{22}'&=
-\frac{2 ( z_{22}-1)^2}{y_{22} z_{22}^2}+
\eta + \theta_0  - (\eta+ 2 \theta_0-\theta_1 e^t) z_{22} + \theta_0 z_{22}^2
.
\end{aligned}
$$
In this chart, the only visible base point on the exceptional line $\mathcal{L}_1$ is $(y_{22}=0,z_{22}=1)$, which is $a_4$.

The energy (\ref{eq:E}) and its rate of change are:
$$
\begin{aligned}
E&=
\frac{1}{y_{22}^2 z_{22}^5}-\frac{2}{y_{22}^2 z_{22}^4}+\frac{1}{y_{22}^2 z_{22}^3}
-\frac{\eta+\theta_0}{y_{22} z_{22}^3}+\frac{\eta+2 \theta_0}{y_{22} z_{22}^2}-\frac{\theta_0}{y_{22} z_{22}}+\frac{F}{z_{22}}
,
\\
E'&=
\frac{\theta_1 e^t}{y_{22}^2 z_{22}^5}
\left(
-F y_{22}^2 z_{22}^4-\theta_0 y_{22} z_{22}^4+(\eta +\theta_0) y_{22} z_{22}^2+z_{22}^2-1
\right)
. 
\end{aligned}
$$

\subsection{Resolution at $a_2$}\label{a2-blow}
\paragraph{First chart:}
\begin{gather*}
y_{31}=\frac{y_{12}}{z_{12}}=\frac{1}{yz},
\quad
z_{31}=z_{12}=z,
\\
y=\frac{1}{y_{31}z_{31}},
\quad
z=z_{31}.
\end{gather*}
The exceptional line is $\mathcal{L}_2 : z_{31}=0$, while the proper preimage of the line $\mathcal{L}_{0}$ is $y_{31}=0$. 

The vector field (\ref{eq:PVlog-system}) is:
$$
\begin{aligned}
y_{31}'&=\frac{1 - (\eta + \theta_0) y_{31} + F y_{31}^2}{y_{31} z_{31}}-y_{31} z_{31} + \theta_0 y_{31}^2 z_{31}
,
\\
z_{31}'&=
- \frac{3}{y_{31}^2} + \frac{2( \eta+ \theta_0)}{y_{31}}+ \frac{ 4 z_{31}}{y_{31}} - (\eta + 2 \theta_0-\theta_1 e^t) z_{31}  - z_{31}^2-F 
.
\end{aligned}
$$
There are two base points on the exceptional line $\mathcal{L}_2$; their $y_{31}$ coordinates are the solutions of the quadratic equation:
$$
1 - (\eta + \theta_0) y_{31} + F y_{31}^2=0.
$$
Thus, the base points are
$$
a_5\left(y_{31}=\frac{2}{\theta_0+\eta+\theta_{\infty}},z_{31}=0\right),
\quad
a_6\left(y_{31}=\frac{2}{\theta_0+\eta-\theta_{\infty}},z_{31}=0\right).
$$

The energy (\ref{eq:E}) and its rate of change are:
$$
\begin{aligned}
E&=
\frac{1}{y_{31}^3 z_{31}}
-\frac{\eta+\theta_0}{y_{31}^2 z_{31}}
-\frac{2}{y_{31}^2}
+\frac{F}{y_{31} z_{31}}+\frac{2 \theta_0+\eta+z_{31}}{y_{31}}-\theta_0 z_{31}
,\\
E'&=
-\theta_1 e^t
\left(
\frac{1}{y_{31}^3 z_{31}}
-\frac{\eta+\theta_0}{y_{31}^2 z_{31}}+\frac{F}{y_{31} z_{31}}-\frac{z_{31}}{y_{31}}+\theta_0 z_{31}
\right)
.
\end{aligned}
$$

\paragraph{Second chart:}
\begin{gather*}
y_{32}=y_{12}=\frac1{y},
\quad
z_{32}=\frac{z_{12}}{y_{12}}=yz,
\\
y=\frac1{y_{32}},
\quad
z=y_{32}z_{32}.
\end{gather*}
The exceptional line is $\mathcal{L}_2 : y_{32}=0$, while the proper preimage of the line $\mathcal{L}_{0}$ is not visible in this chart. 

The vector field (\ref{eq:PVlog-system}) is:
$$
\begin{aligned}
y_{32}'&=
\eta + \theta_0 - (\eta + 2 \theta_0-\theta_1 e^t) y_{32} + \theta_0 y_{32}^2 - 2 z_{32} + 
 4 y_{32} z_{32} - 2 y_{32}^2 z_{32}
,
\\
z_{32}'&=
-\frac{F-(\eta+\theta_0) z_{32}+z_{32}^2}{y_{32}} - \theta_0 y_{32} z_{32} + y_{32} z_{32}^2
.
\end{aligned}
$$
In this chart, the only base points on the exceptional line $\mathcal{L}_2$ are $a_5$ and $a_6$.

The energy (\ref{eq:E}) and its rate of change are:
$$
\begin{aligned}
E&=
\frac{F-(\eta+\theta_0) z_{32}+z_{32}^2}{y_{32}}
+(\eta+ 2 \theta_0)z_{32}-\theta_0 y_{32} z_{32}+y_{32} z_{32}^2-2 z_{32}^2
,
\\
E'&=
-\theta_1 e^t
\left(
\frac{F-(\eta+\theta_0) z_{32}+z_{32}^2}{y_{32}}
+\theta_0 y_{32} z_{32}-y_{32} z_{32}^2
\right)
. 
\end{aligned}
$$

\subsection{Resolution at $a_3$}
\paragraph{First chart:}
\begin{gather*}
y_{41}=\frac{y_{21}}{z_{21}}=yz,
\quad
z_{41}=z_{21}=\frac1z,
\\
y=y_{41}z_{41},
\quad
z=\frac{1}{z_{41}}.
\end{gather*}
The exceptional line is $\mathcal{L}_3 : z_{41}=0$, while the proper preimage of the line $\mathcal{L}_{1}$ is not visible in this chart. 

The vector field  (\ref{eq:PVlog-system}) is:
$$
\begin{aligned}
y_{41}'&=
\frac{y_{41}-\theta_0}{z_{41}}  - F y_{41} z_{41} + \eta y_{41}^2 z_{41} + \theta_0 y_{41}^2 z_{41} - 
 y_{41}^3 z_{41}
 ,
\\
z_{41}'&= 
1 + (\eta + 2 \theta_0-\theta_1 e^t ) z_{41} - 4 y_{41} z_{41} + F z_{41}^2 - 2 \eta y_{41} z_{41}^2 - 
 2 \theta_0 y_{41} z_{41}^2 + 3 y_{41}^2 z_{41}^2
.
\end{aligned}
$$
In this chart, there is one base point on the exceptional line $\mathcal{L}_3$:
$$
a_7(y_{41}=\theta_0,z_{41}=0).
$$

The energy (\ref{eq:E}) is:
$$
E=
\frac{y_{41}-\theta_0}{z_{41}}
+y_{41}^3 z_{41}
-(\eta+\theta_0) y_{41}^2 z_{41}
-2 y_{41}^2
+F y_{41} z_{41}+(2 \theta_0+\eta) y_{41}.
$$

\paragraph{Second chart:}
\begin{gather*}
y_{42}=y_{21}=y,
\quad
z_{42}=\frac{z_{21}}{y_{21}}=\frac1{yz},
\\
y=y_{42},
\quad
z=\frac1{y_{42}z_{42}}.
\end{gather*}
The exceptional line is $\mathcal{L}_3 : y_{42}=0$, while the proper preimage of the line $\mathcal{L}_{1}$ is $z_{42}=0$. 

The vector field  (\ref{eq:PVlog-system}) is:
$$
\begin{aligned}
y_{42}'&=
\frac{ 2 ( y_{42}-1)^2}{z_{42}}
-\theta_0 + (\eta + 2 \theta_0-\theta_1 e^t) y_{42} - (\eta + \theta_0) y_{42}^2
,
\\
z_{42}'&=
\frac{\theta_0 z_{42}-1}{y_{42}}+ y_{42}  - (\eta + \theta_0) y_{42} z_{42} + 
 F y_{42} z_{42}^2
.
\end{aligned}
$$
In this chart, the only base point on the exceptional line $\mathcal{L}_3$ is $a_7(y_{42}=0,z_{42}=1/\theta_0)$.

The energy (\ref{eq:E}) is:
$$
E=
\frac{1}{y_{42} z_{42}^2}
+\frac{y_{42}-2}{z_{42}^2}-\frac{\theta_0}{y_{42} z_{42}}
+\frac{2 \theta_0+\eta-(\theta_0+\eta)y_{42}}{z_{42}}
+F y_{42}
.
$$

\subsection{Resolution at $a_4$}
\paragraph{First chart:}
\begin{gather*}
y_{51}=\frac{y_{21}-1}{z_{21}}=(y-1)z,
\quad
z_{51}=z_{21}=\frac1z,
\\
y=y_{51}z_{51}+1,
\quad
z=\frac1{z_{51}}.
\end{gather*}
The exceptional line is $\mathcal{L}_4 : z_{51}=0$, while the proper preimage of the line $\mathcal{L}_{1}$ is not visible in this chart. 

The vector field  (\ref{eq:PVlog-system}) is:
$$
\begin{aligned}
y_{51}'&=
-\frac{\theta_1 e^t}{z_{51}}
-F y_{51} z_{51} + (\eta + \theta_0) y_{51}^2 z_{51} - y_{51}^3 z_{51}
,
\\
z_{51}'&=
-(\eta+\theta_1 e^t) z_{51}+ 2 y_{51} z_{51} + F z_{51}^2 - 2 (\eta + \theta_0) y_{51} z_{51}^2 + 3 y_{51}^2 z_{51}^2
.
\end{aligned}
$$
There are no base points on the exceptional line $\mathcal{L}_4$ in this chart.

The energy (\ref{eq:E}) is:
$$
E=
y_{51}^3 z_{51}
-(\eta+\theta_0) y_{51}^2 z_{51}
+y_{51}^2
+F y_{51} z_{51}
-\eta y_{51}+F
.
$$

\paragraph{Second chart:}
\begin{gather*}
y_{52}=y_{21}-1=y-1,
\quad
z_{52}=\frac{z_{21}}{y_{21}-1}=\frac1{(y-1)z},
\\
y=y_{52}+1,
\quad
z=\frac{1}{y_{52}z_{52}}.
\end{gather*}
The exceptional line is $\mathcal{L}_4 : y_{52}=0$, while the proper preimage of the line $\mathcal{L}_{1}$ is $z_{52}=0$. 

The vector field  (\ref{eq:PVlog-system}) is:
$$
\begin{aligned}
y_{52}'&=
-\theta_0 + (\eta + 2 \theta_0-\theta_1 e^t) (1 + y_{52}) - (\eta + \theta_0) (1 + y_{52})^2 + 
\frac{ 2 y_{52} (1 + y_{52})}{z_{52}}
 ,
\\
z_{52}'&=\frac{\theta_1 e^t  z_{52}}{y_{52}}+
y_{52} - (\eta + \theta_0) y_{52} z_{52} + F y_{52} z_{52}^2
.
\end{aligned}
$$
In this chart, there is one base point on the exceptional line $\mathcal{L}_4$:
$$
a_8(y_{52}=0,z_{52}=0).
$$
Notice that $a_8$ is the intersection point of $\mathcal{L}_4$ and the proper preimage of  $\mathcal{L}_{1}$.

The energy (\ref{eq:E}) is:
$$
E=
\frac{y_{52}+1}{z_{52}^2}
-\frac{(\eta+\theta_0) y_{52}+\eta}{z_{52}}+F y_{52}+F
.
$$

\subsection{Resolution at $a_5$}\label{a5-blow}

\paragraph{First chart:}
\begin{gather*}
\tilde{y}_{61}=\frac{y_{32}}{z_{32}-\epsilon}=\frac{1}{y(yz-\epsilon)},
\quad
\tilde{z}_{61}=z_{32}-\epsilon=yz-\epsilon,
\\
y=\frac{1}{\tilde{y}_{61}\tilde{z}_{61}},
\quad
z=\tilde{y}_{61}\tilde{z}_{61}(\tilde{z}_{61}+\epsilon).
\end{gather*}
The exceptional line is $\mathcal{L}_5 : z_{61}=0$, and the proper preimage of $\mathcal{L}_{2}$ is $y_{61}=0$.
$\mathcal{L}_0$ is not visible in this chart. 

The vector field  (\ref{eq:PVlog-system}) is:
$$
\begin{aligned}
y_{61}'=&
-1 + (4 \epsilon - \eta  - 2 \theta_0  + \theta_1 e^t) y_{61} 
+\epsilon( \epsilon -  \theta_0) y_{61}^2 + 4 y_{61} z_{61} +2(\theta_0- 2\epsilon) y_{61}^2 z_{61}-  3 y_{61}^2 z_{61}^2
,
\\
z_{61}'=&
\frac{\eta+\theta_0-2\epsilon-z_{61}}{y_{61}}
 + \epsilon (\epsilon-\theta_0)y_{61} z_{61} + (2 \epsilon - \theta_0) y_{61} z_{61}^2 + y_{61} z_{61}^3
.
\end{aligned}
$$
In this chart, there are no base points on the exceptional line $\mathcal{L}_5$.

The energy (\ref{eq:E}) is:
$$
\begin{aligned}
E=& 
\frac{2 \epsilon-\eta-\theta_0+z_{61}}{y_{61}}
+\epsilon(\eta+2\theta_0-2\epsilon)
+(\eta-4 \epsilon+2 \theta_0) z_{61}
-2 z_{61}^2
\\
&
+\epsilon(\epsilon-\theta_0) y_{61} z_{61}
+(2 \epsilon-\theta_0)y_{61} z_{61}^2
+y_{61} z_{61}^3
.
\end{aligned}
$$

\paragraph{Second chart:}
\begin{gather*}
\tilde{y}_{62}=y_{32}=\frac{1}{y},
\quad
\tilde{z}_{62}=\frac{z_{32}-\epsilon}{y_{32}}=y(yz-\epsilon),
\\
y=\frac{1}{\tilde{y}_{62}},
\quad
z=\tilde{y}_{62}(y_{62}\tilde{z}_{62}+\epsilon).
\end{gather*}
The exceptional line is $\mathcal{L}_5\ :\ y_{62}=0$, while the proper preimages of $\mathcal{L}_2$ and $\mathcal{L}_0$ are not visible in this chart.

The vector field  (\ref{eq:PVlog-system}) is:
$$
\begin{aligned}
y_{62}'=&
\eta-2 \epsilon+\theta_0
+(4 \epsilon-\eta-2 \theta_0+\theta_1 e^t ) y_{62}
+(\theta_0-2 \epsilon) y_{62}^2
-2 y_{62} z_{62}
+4 y_{62}^2 z_{62}
-2 y_{62}^3 z_{62}
,\\
z_{62}'=&
\epsilon(\epsilon- \theta_0)
+(\eta +2 \theta_0 -4 \epsilon -\theta_1 e^t )z_{62}
+z_{62}^2
+2(2 \epsilon -\theta_0) y_{62} z_{62}
-4 y_{62} z_{62}^2
+3 y_{62}^2 z_{62}^2
.
\end{aligned}
$$
In this chart, there are no base points on the exceptional line $\mathcal{L}_5$.

The energy (\ref{eq:E}) is:
$$
\begin{aligned}
E=&
\epsilon(\eta-2\epsilon+2  \theta_0)
+\epsilon(\epsilon- \theta_0) y_{62}
+(2 \epsilon-\eta -\theta_0 ) z_{62}
\\&
+(\eta -4 \epsilon+2 \theta_0 ) y_{62} z_{62}
+(2 \epsilon -\theta_0) y_{62}^2 z_{62}
+y_{62} z_{62}^2
-2 y_{62}^2 z_{62}^2
+y_{62}^3 z_{62}^2
.
\end{aligned}
$$

\subsection{Resolution at $a_6$}\label{a6-blow}

Same as at $a_5$, replacing $\theta_{\infty}$ with $-\theta_{\infty}$.
The resolution does not yield new base points.

\subsection{Resolution at $a_7$}\label{a7-blow}
\paragraph{First chart:}
\begin{gather*}
y_{81}=\frac{y_{41}-\theta_0}{z_{41}}=z(yz-\theta_0),
\quad
z_{81}=z_{41}=\frac1z,
\\
y=z_{81}(y_{81}z_{81}+\theta_0),
\quad
z=\frac1{z_{81}}.
\end{gather*}
The exceptional line is $\mathcal{L}_7 : z_{81}=0$, while the proper preimages of the lines $\mathcal{L}_{3}$ and $\mathcal{L}_1$ are not visible in this chart. 

The vector field  (\ref{eq:PVlog-system}) is:
$$
\begin{aligned}
y_{81}'=&
 \theta_0(\eta \theta_0-F) + (2 \theta_0-\eta+\theta_0 e^t) y_{81}
+2( 2 \eta \theta_0  -  \theta_0^2-F) y_{81} z_{81} 
\\&
+ 4 y_{81}^2 z_{81} + 
 3( \eta - 2 \theta_0 )y_{81}^2 z_{81}^2 - 4 y_{81}^3 z_{81}^3
,
\\
z_{81}'=&
1 + (\eta - 2 \theta_0-\theta_1 e^t) z_{81} + (F  - 2 \eta \theta_0  +  \theta_0^2) z_{81}^2 - 4 y_{81} z_{81}^2
+2(2\theta_0-  \eta) y_{81} z_{81}^3  + 
 3 y_{81}^2 z_{81}^4
.
\end{aligned}
$$
There are no base points on $\mathcal{L}_7$ in this chart.

The energy (\ref{eq:E}) is:
$$
\begin{aligned}
E=&
\eta \theta_0
+y_{81}
+ \theta_0(F-\eta \theta_0) z_{81}
+(\eta-2 \theta_0) y_{81} z_{81}
\\&
+(F +\theta_0^2 -2 \eta \theta_0) y_{81} z_{81}^2
-2 y_{81}^2 z_{81}^2
+(2 \theta_0 -\eta) y_{81}^2 z_{81}^3
+y_{81}^3 z_{81}^4
.
\end{aligned}
$$

\paragraph{Second chart:}
\begin{gather*}
y_{82}=y_{41}-\theta_0=yz-\theta_0,
\quad
z_{82}=\frac{z_{41}}{y_{41}-\theta_0}=\frac{1}{z(yz-\theta_0)},
\\
y=y_{82}z_{82}(y_{82}+\theta_0),
\quad
z=\frac1{y_{82}z_{82}}.
\end{gather*}
The Jacobian and its derivative are:
$$
\begin{aligned}
J_{82}=& \frac{\partial y_{82}}{\partial y}\frac{\partial z_{82}}{\partial z}-\frac{\partial y_{82}}{\partial z}\frac{\partial z_{82}}{\partial y}
=\frac{1}{z(\theta_0 -yz)}=-z_{82},
\\
J_{82}'=& -z_{82}(
\eta - 2 \theta_0 - \theta_1 e^t- 4 y_{82} + 
\theta_0(F - \eta \theta_0) z_{82} 
+ 2( F-2 \eta \theta_0 + \theta_0^2) y_{82} z_{82} 
\\
&\qquad\quad
+3(2\theta_0- \eta) y_{82}^2 z_{82} + 4 y_{82}^3 z_{82}
 ).
\end{aligned}
$$
The exceptional line is $\mathcal{L}_7 : y_{82}=0$, while the proper preimage of the line $\mathcal{L}_{3}$ is $z_{82}=0$. 
$\mathcal{L}_1$ is not visible in this chart.

The vector field  (\ref{eq:PVlog-system}) is:
$$
\begin{aligned}
y_{82}'=&
\frac{1}{z_{82}} 
+\theta_0(\eta\theta_0-F) y_{82} z_{82}
+(  2 \eta \theta_0  - \theta_0^2-F) y_{82}^2 z_{82} + (\eta  -  2 \theta_0) y_{82}^3 z_{82} - y_{82}^4 z_{82}
,
\\
z_{82}'=&
z_{82}(
\eta - 2 \theta_0 - \theta_1 e^t- 4 y_{82} + 
\theta_0(F - \eta \theta_0) z_{82} 
+ 2( F-2 \eta \theta_0 + \theta_0^2) y_{82} z_{82} 
\\
&\qquad
+3(2\theta_0- \eta) y_{82}^2 z_{82} + 4 y_{82}^3 z_{82}
 )
.
\end{aligned}
$$
In this chart, no base points are visible on the exceptional line $\mathcal{L}_7$.

The energy (\ref{eq:E}) is:
$$
\begin{aligned}
E=&
\frac{1}{z_{82}}
+\eta \theta_0
+(\eta -2 \theta_0) y_{82}
-2 y_{82}^2
+\theta_0(F-\eta \theta_0) y_{82} z_{82}
\\&
+(F-2 \eta \theta_0 +\theta_0^2) y_{82}^2 z_{82}
+(2 \theta_0 -\eta) y_{82}^3 z_{82}
+y_{82}^4 z_{82}
.
\end{aligned}
$$

\subsection{Resolution at $a_8$}\label{a8-blow}
\paragraph{First chart:}
\begin{gather*}
y_{91}=\frac{y_{52}}{z_{52}}=(y-1)^2z,
\quad
z_{91}=z_{52}=\frac1{(y-1)z},
\\
y=y_{91}z_{91}+1,
\quad
z=\frac1{y_{91}z_{91}^2}.
\end{gather*}
The exceptional line is $\mathcal{L}_8 : z_{91}=0$, while the proper preimage of the line $\mathcal{L}_{4}$ is $y_{91}=0$.
$\mathcal{L}_{1}$ is not visible in this chart and it corresponds to the infinite value of $y_{91}$.

The vector field  (\ref{eq:PVlog-system}) is:
$$
\begin{aligned}
y_{91}'=&
 \frac{2 (y_{91}-\theta_1 e^t)}{z_{91}}-(\eta+\theta_1 e^t) y_{91} + y_{91}^2  - F y_{91}^2 z_{91}^2
,
\\
z_{91}'=&
\frac{\theta_1 e^t}{y_{91}}+
y_{91} z_{91} (1 - (\eta+\theta_0) z_{91} + F z_{91}^2)
.
\end{aligned}
$$
In this chart, there is one base point on the exceptional line $\mathcal{L}_8$:
$$
a_9(y_{91}=\theta_1 e^t,z_{91}=0).
$$

The energy (\ref{eq:E}) is:
$$
E=
\frac{1}{z_{91}^2}
+\frac{y_{91}-\eta}{z_{91}}
+F
-(\eta+\theta_0) y_{91}
+F y_{91} z_{91}
.
$$

\paragraph{Second chart:}
\begin{gather*}
y_{92}=y_{52}=y-1,
\quad
z_{92}=\frac{z_{52}}{y_{52}}=\frac{1}{(y-1)^2z},
\\
y=y_{92}+1,
\quad
z=\frac{1}{y_{92}^2z_{92}}.
\end{gather*}
The exceptional line is $\mathcal{L}_8 : y_{92}=0$. The proper preimage of the line $\mathcal{L}_{1}$ is $z_{92}=0$, while the proper preimage of $\mathcal{L}_{4}$ is not visible in this chart. 

The vector field  (\ref{eq:PVlog-system}) is:
$$
\begin{aligned}
y_{92}'=&
\frac{ 2(1+ y_{92})}{z_{92}}-\theta_1 e^t
-(\eta+\theta_1 e^t) y_{92} - (\eta + \theta_0) y_{92}^2 
,
\\
z_{92}'=&
 \frac{2(\theta_1 e^t z_{92}-1)}{y_{92}} + (\eta + \theta_1 e^t) z_{92} + F y_{92}^2 z_{92}^2-1 
.
\end{aligned}
$$
In this chart, there is one base point on the exceptional line $\mathcal{L}_8$: $a_9\left(y_{92}=0,z_{92}=\dfrac{1}{\theta_1 e^t}\right)$.

The energy (\ref{eq:E}) is:
$$
E=
\frac{1}{y_{92}^2 z_{92}^2}
+\frac{1}{y_{92} z_{92}^2}
-\frac{\eta}{y_{92} z_{92}}
-\frac{\eta+\theta_0}{z_{92}}
+F y_{92}
+F
.
$$

\subsection{Resolution at $a_9$}\label{a9-blow}
\paragraph{First chart:}
\begin{gather*}
y_{101}=\frac{y_{91}-\theta_1e^t}{z_{91}}= (y-1)z ( (y-1)^2z - \theta_1 e^t ),
\quad
z_{101}=z_{91}=\frac1{(y-1)z},
\\
y=1 + \theta_1e^t z_{101} + y_{101} z_{101}^2,
\quad
z=\frac{1}{z_{101}^2 (\theta_1e^t + y_{101} z_{101})}.
\end{gather*}
The exceptional line is $\mathcal{L}_9 : z_{101}=0$, while the proper preimage of $\mathcal{L}_{8}$ is not visible in this chart. 

The vector field  (\ref{eq:PVlog-system}) is:
$$
\begin{aligned}
y_{101}'=&
\frac{\theta_1 e^t  (y_{101}-(1 + \eta)\theta_1  e^t )}{z_{101} (\theta_1 e^t  + y_{101} z_{101})}
+
\frac{y_{101} (2y_{101}-(1 + \eta)  \theta_1 e^t)}{\theta_1 e^t  + y_{101} z_{101}}
-\eta y_{101}   - F\theta_1^2 e^{2t} z_{101}
\\
&
+\theta_1(\eta+\theta_0) e^t y_{101} z_{101}
-3F\theta_1 e^ty_{101}z_{101}^2 + (\eta+\theta_0) y_{101}^2 z_{101}^2 - 2 F y_{101}^2 z_{101}^3
,\\
z_{101}'=&
\frac{\theta_1 e^t}{\theta_1 e^t  + y_{101} z_{101}}
+
\theta_1 e^t z_{101} - (\eta +\theta_0) \theta_1 e^t z_{101}^2 + 
 y_{101} z_{101}^2 +  F\theta_1 e^t z_{101}^3
\\& 
   - (\eta +\theta_0) y_{101} z_{101}^3 + F y_{101} z_{101}^4
 .
\end{aligned}
$$
In this chart, there is one base point on the exceptional line $\mathcal{L}_9$:
$$
a_{10}\left(y_{101}=(1 + \eta)\theta_1  e^t,z_{101}=0\right).
$$

The energy (\ref{eq:E}):
$$
E=
\frac{1}{z_{101}^2}
+\frac{\theta_1 e^t-\eta}{z_{101}}
-(\eta+\theta_0) \theta_1 e^t
+F
+y_{101}
+\theta_1 e^t F  z_{101}
-(\eta+\theta_0) y_{101} z_{101}
+F y_{101} z_{101}^2
.
$$

\paragraph{Second chart:}
\begin{gather*}
y_{102}=y_{91}-\theta_1 e^t=(y-1)^2z-\theta_1 e^t,
\quad
z_{102}=\frac{z_{91}}{y_{91}-\theta_1 e^t}=\frac{1}{(y-1)z}\cdot\frac{1}{(y-1)^2z-\theta_1 e^t}
,
\\
y= 1 + \theta_1 e^t y_{102} z_{102} + y_{102}^2 z_{102},
\quad
z=\frac{1}{y_{102}^2 (\theta_1 e^t + y_{102}) z_{102}^2}.
\end{gather*}
The Jacobian is:
$$
\begin{aligned}
J_{102}=&\frac{\partial y_{102}}{\partial y}\frac{\partial z_{102}}{\partial z}-\frac{\partial y_{102}}{\partial z}\frac{\partial z_{102}}{\partial y}
=\frac{1}{z (\theta_1 e^t - ( y-1)^2 z)}
=-y_{102} (\theta_1 e^t + y_{102}) z_{102}^2,
\\
J_{102}'=&
-2 y_{102}^3 z_{102}^2\big(
1- (\eta  +  \theta_0) y_{102} z_{102} 
+  F y_{102}^2 z_{102}^2 
\big)
- 
 F \theta_1^3 e^{3t}  y_{102}^2 z_{102}^4 
\\& 
 - \theta_1^2  e^{2t} z_{102}^2 \big(1+ \eta +y_{102}-   2 (\eta + \theta_0)  y_{102}^2 z_{102} 
 +   4 F  y_{102}^3 z_{102}^2\big) 
\\& 
 -
\theta_1 e^{t} y_{102} z_{102}^2  \big(1 + \eta +   3  y_{102} - 4 (\eta+\theta_0)  y_{102}^2 z_{102} + 5 F  y_{102}^3 z_{102}^2\big).
\end{aligned}
$$
The exceptional line is $\mathcal{L}_9 : y_{102}=0$.
The proper preimages of $\mathcal{L}_{8}$ is $z_{102}=0$ and of $\mathcal{L}_4$ is $y_{102}=-\theta_1 e^t$.
 $\mathcal{L}_1$ is not visible in this chart and it corresponds the infinite value of $y_{102}$.

The vector field  (\ref{eq:PVlog-system}) is:
$$
\begin{aligned}
y_{102}'=&
\frac{2}{z_{102}}
-\eta y_{102} + y_{102}^2  - F y_{102}^4 z_{102}^2
-(1+ \eta) \theta_1 e^t  
\\&
+ \theta_1 e^t  y_{102} - 
F \theta_1^2 e^{2t}  y_{102}^2 z_{102}^2 - 2 F \theta_1 e^t  y_{102}^3 z_{102}^2
,
\\
z_{102}'=&
\frac{\theta_1 e^t ((1 + \eta) \theta_1 e^t z_{102}-1)}{y_{102} (\theta_1 e^t + y_{102})}
+\frac{ (1+\eta) \theta_1 e^t z_{102}-2}{\theta_1 e^t + y_{102}}
\\&
-\theta_1 e^t + (1+\eta)\theta_1^2 e^{2t}  z_{102}  -2y_{102}
   + (1+ \eta )\theta_1 e^t y_{102}z_{102}
.
\end{aligned}
$$
In this chart, there is one visible base point on the exceptional line $\mathcal{L}_9$ : 
$$a_{10} \left( y_{102}=0,z_{102}=\frac{1}{(1 + \eta)\theta_1  e^t}\right).$$

The energy (\ref{eq:E}) is:
$$
\begin{aligned}
E=&
\frac{1}{y_{102}^2 z_{102}^2}
+\frac{\theta_1 e^t-\eta}{y_{102} z_{102}}
+\frac{1}{z_{102}}
\\&
+F
-(\eta+\theta_0) \theta_1 e^t
-(\eta+\theta_0) y_{102}
+e^tF \theta_1 y_{102} z_{102}
+F y_{102}^2 z_{102}
.
\end{aligned}
$$

\subsection{Resolution at $a_{10}$}\label{a10-blow}
\paragraph{First chart:}
$$
\begin{aligned}
y_{111}&=\frac{y_{101}-(1 + \eta)\theta_1  e^t}{z_{101}}
= (y-1)^4z^3 - \theta_1 e^t(y-1)^2z^2  -(1 + \eta)\theta_1  e^t   (y-1)z
,
\\
z_{111}&=z_{101}=\frac1{(y-1)z},
\\
y &= 1 + \theta_1 e^t z_{111} + \theta_1 e^t z_{111}^2 + \eta \theta_1 e^t z_{111}^2 + y_{111} z_{111}^3,
\\
z &= \frac{1}{ z_{111}^2 (\theta_1 e^t + \theta_1 e^t z_{111} + \eta \theta_1 e^t z_{111} + y_{111} z_{111}^2)}.
\end{aligned}
$$
The exceptional line is $\mathcal{L}_{10} : z_{111}=0$, while the proper preimage of $\mathcal{L}_{9}$ is not visible in this chart. 

The vector field  (\ref{eq:PVlog-system}) is:
$$
\begin{aligned}
y_{111}'=&\frac{1}{y_{111} z_{111}^2 + \theta_1 e^t (1 + (1+\eta) z_{111})}
\times
\\&
\times
\big(
(e^{3t} \theta_1^3 (1 + (1+\eta) z_{111})^2 (\eta^2 + \theta_0 - 
     F (1 + 2 z_{111}) + \eta (1 + \theta_0 - 2 F z_{111})) 
\\&
\qquad     
     - 
  y_{111}^2 z_{111} ( y_{111} z_{111}^3 (1 - 2 \theta_0 z_{111} + 3 F z_{111}^2) + 
     \eta (z_{111} - 2 y_{111} z_{111}^4)-2) 
 \\&
\qquad      
     + 
  \theta_1 e^t y_{111} (3 - z_{111} - 2 y_{111} z_{111}^2 
  + 
     2 (2 \theta_0-1) y_{111} z_{111}^3 + ( 5 \theta_0 y_{111}-7 F y_{111}) z_{111}^4)
\\&
 \qquad
+ \theta_1 e^t y_{111} (
         \eta^2 z_{111} (5 y_{111} z_{111}^3-2)  - 
     8 F y_{111} z_{111}^5)
\\&
 \qquad
+\eta \theta_1 e^t y_{111} (
     2 - 3 z_{111} + 2 y_{111} z_{111}^3 + 5 (1 + \theta_0) y_{111} z_{111}^4 - 
        8 F y_{111} z_{111}^5) 
\\&
 \qquad        
        -
  e^{2t} \theta_1^2 (1+ 
     y_{111} (1 + z_{111}) (1 + z_{111} - 
        2 \theta_0 z_{111} + (5 F - 4 \theta_0) z_{111}^2 + 7 F z_{111}^3))
\\&
 \qquad        
       +
 \eta  \theta_1^2 e^{2t}(
        2 y_{111} z_{111}^2 (2 + 3 \theta_0 + 2(1-3F+2 \theta_0+2\eta^2) z_{111} - 
           7 F z_{111}^2)-3-\eta^2)
 \\&
 \qquad        
       +
 \eta^2\theta_1^2 e^{2t}                   
      ( y_{111} z_{111}^2 (5 + 4 (2 + \theta_0) z_{111} - 7 F z_{111}^2)-3)
\big)
 ,\\
z_{111}'=&
\frac{1}{y_{111} z_{111}^2 + \theta_1 e^t (1 + z_{111} + \eta z_{111})}
\times
\\&
\times
\big(
y_{111}^2 z_{111}^5 (1 - \eta z_{111} - \theta_0 z_{111} + F z_{111}^2) 
\\
&
\qquad
+ 
 e^{2t} \theta_1^2 z_{111} (1 + z_{111} + \eta z_{111})^2 (1 - \eta z_{111} - 
    \theta_0 z_{111} + F z_{111}^2) 
\\
&
\qquad
        + 
 \theta_1 e^t (1 + 
    2 y_{111} z_{111}^3 (1 + z_{111} + \eta z_{111}) (1 - \eta z_{111} - \theta_0 z_{111} + 
       F z_{111}^2))
\big)
.
\end{aligned}
$$
In this chart, there are no base points on the exceptional line $\mathcal{L}_{10}$.

The energy (\ref{eq:E}) is:
$$
\begin{aligned}
E= &
\frac{1}{z_{111}^2}
+\frac{\theta_1 e^t-\eta}{z_{111}}
+F+(1-\theta_0) \theta_1e^t
+e^t \theta_1(F-(1+\eta)(\eta+\theta_0)) z_{111}
\\&
+(1+\eta)\eta e^tF \theta_1 z_{111}^2
+y_{111} z_{111}
-(\eta+ \theta_0) y_{111} z_{111}^2
+F y_{111} z_{111}^3
.
\end{aligned}
$$

\paragraph{Second chart:}
$$
\begin{aligned}
y_{112}=&y_{101}-(1 + \eta)\theta_1  e^t
= (y-1)^3z^2 - (y-1)z \theta_1 e^t -(1 + \eta)\theta_1  e^t
,
\\
z_{112}=&\frac{z_{101}}{y_{101}-(1 + \eta)\theta_1  e^t}
=\frac{1}{(y-1)^4z^3 - (y-1)^2z^2 \theta_1 e^t -(y-1)z(1 + \eta)\theta_1  e^t}
,
\\
y=&1 +  \theta_1 e^t y_{112} z_{112} + (1 + \eta)\theta_1  e^t y_{112}^2 z_{112}^2 + y_{112}^3 z_{112}^2,
\\
z=&\frac1{\theta_1  e^t y_{112}^2 z_{112}^2 (1 + (1 + \eta) y_{112} z_{112} + y_{112}^2 z_{112})}.
\end{aligned}
$$
The Jacobian and its derivative are:
$$
\begin{aligned}
J_{112} = &
-\theta_1 e^t z_{112} -(1+\eta) \theta_1 e^t y_{112} z_{112}^2-y_{112}^2 z_{112}^2,
\\
J_{112}' = &
2 y_{112}^4 z_{112}^3\big( (\eta+\theta_0) y_{112} z_{112} - 1 +   F y_{112}^2 z_{112}^2\big)
\\&
+\theta_1 e^t z_{112}\big(
1 + \eta - (1+\eta)^2y_{112} z_{112}  - 
 3 y_{112}^2 z_{112} + 4 (\theta_0-1) y_{112}^3 z_{112}^2 \big.
 \\&\qquad\qquad\quad
 \big.
 + 
 5 (\eta+ \eta^2 - F +\theta_0+ \eta \theta_0) y_{112}^4 z_{112}^3 
 - 6 F(1+ \eta) y_{112}^5 z_{112}^4
\big)
\\&
+\theta_1^2 e^{2t} z_{112}\big(
1 + (1+\eta)^3 z_{112} 
+ (3 +  \eta  - 2 \theta_0) y_{112} z_{112} 
\big.
\\&\qquad\qquad\quad
\big.
+ 2 (1 - \eta - 2 \eta^2  + 2 F  - 3 \theta_0(1+\eta)) y_{112}^2 z_{112}^2
\big.
\\&\qquad\qquad\quad
\big.
 - 2(1+\eta)( 2(\eta+\theta_0) (1+\eta) - 5 F ) y_{112}^3 z_{112}^3 + 
 6 F (1+\eta)^2y_{112}^4 z_{112}^4
\big)
.
\end{aligned}
$$
The exceptional line is $\mathcal{L}_{10} : y_{112}=0$, while the proper preimage of $\mathcal{L}_{9}$ is $z_{112}=0$. 
$\mathcal{L}_{8}$ is not visible in this chart.

The vector field  (\ref{eq:PVlog-system}) is:
$$
\begin{aligned}
y_{112}'=&
\frac{2 y_{112}^2 z_{112} + \theta_1 e^t (1 + 2 (1 + \eta) y_{112} z_{112})}{z_{112} (y_{112}^2 z_{112} + \theta_1 e^t (1 + (1 + \eta) y_{112} z_{112}))}
+
\frac{(1 + \eta) \theta_1 e^t ((1 + \eta) \theta_1 e^t + y_{112})}
{y_{112}^2 z_{112} + \theta_1 e^t (1 + (1 + \eta) y_{112} z_{112})}
\\&
-(1+\eta)^2\theta_1 e^t   - 
 2 F y_{112}^5 z_{112}^3 -\eta  y_{112}+
\theta_1^2 e^{2t} y_{112}z_{112} ( \eta   + \eta^2    - 
     F   +  \theta_0   + \eta    \theta_0  ) 
\\&     
     + 
 \theta_1 e^t  y_{112}^2 z_{112}(\eta  + \theta_0  ) + 
e^{2t} \theta_1^2 (1+\eta) y_{112}^2 z_{112}^2(    \eta   + 2 \eta^2  - 3  F     + \theta_0  + \eta \theta_0   ) 
\\&    
    + 
(\eta  + \theta_0) y_{112}^4 z_{112}^2 -4 e^t F \theta_1(1+\eta) y_{112}^4 z_{112}^3
 +\theta_1 e^t y_{112}^3 z_{112}^2(2 \eta  + 2 \eta^2  -     3  F + 2  \theta_0   +     2 \eta  \theta_0 )     
\\&     
    -2 e^{2t} F \theta_1^2(1  +   \eta)^2  y_{112}^3z_{112}^3
,
\\
z_{112}'=&
\frac{1}{y_{112}^2 z_{112} + \theta_1 e^t  (1 + (1 + \eta) y_{112} z_{112})}
\times
\\
&\qquad\times
\big(
{-}3 (1+\eta)\theta_1 e^t  z_{112} + (1+\eta)^3 \theta_1^2  e^{2t} z_{112}^2  
-2  y_{112} z_{112} + (1+\eta)^2\theta_1 e^t y_{112}   z_{112}^2 
\big)
\\&
+
(\eta  + \theta_1 e^t)  z_{112} 
- e^{2t} \theta_1^2(\eta  +  \eta^2  -  F   +   \theta_0   + \eta  \theta_0 ) z_{112}^2 
\\&
+ \theta_1 e^t  y_{112}  z_{112}^2(1  - \eta  - 2  \theta_0  )
-(1 + \eta) e^{2t} (\eta + \eta^2 - 3 F + \theta_0 +  \eta \theta_0) \theta_1^2 y_{112} z_{112}^3
\\&
 + y_{112}^2z_{112}^2
 +\theta_1 e^t  y_{112}^2z_{112}^3 (  4 F -  3 (\eta+\theta_0)(1+\eta)) 
    +2 (1+\eta)^2e^{2t} F \theta_1^2 y_{112}^2  z_{112}^4
\\& 
  -2(\eta+\theta_0) y_{112}^3 z_{112}^3  
 + 5F \theta_1( 1  +    \eta  ) e^{t} y_{112}^3z_{112}^4 
 + 3 F y_{112}^4 z_{112}^4 
.
\end{aligned}
$$
There are no base points on the exceptional line $\mathcal{L}_{10}$ in this chart.

The energy (\ref{eq:E}) is:
$$
\begin{aligned}
E=&
\frac{1}{y_{112}^2 z_{112}^2}
+\frac{\theta_1 e^t-\eta}{y_{112} z_{112}}
+F +\theta_1 e^t(1-  \theta_0)
+y_{112}
+\theta_1 e^t(F -(1+\eta)(\eta+\theta_0))y_{112} z_{112}
\\
&
-(\eta+\theta_0) y_{112}^2 z_{112}
+(1+\eta) e^{t} F \theta_1 y_{112}^2 z_{112}^2
+F y_{112}^3 z_{112}^2
.
\end{aligned}
$$

\section{Notation}
\label{sec:notation}
In this appendix, we collect the notation for base points, provide the charts in which they are defined and their coordinates in these charts, and give the relationships between constants used in the paper. 

\begin{longtable}{|c|l|c|}
\hline
base point & coordinate system & coordinates\\
\hline
  $a_0$ & $(y_{02},z_{02})=(\frac1y,\frac{z}{y})$ & $(0,0)$ 
  \\
  $a_1$ & $(y_{03},z_{03})=(\frac{y}{z},\frac1{z})$ & $(0,0)$  
  \\
  $a_2$ & $(y_{12},z_{12})=(y_{02},\frac{z_{02}}{y_{02}})=(\frac1y,z)$ & $(0,0)$
  \\
  $a_3$ & $(y_{21},z_{21})=(\frac{y_{03}}{z_{03}},z_{03})=(y,\frac1z)$ & $(0,0)$
  \\  
  $a_4$ & $(y_{21},z_{21})=(\frac{y_{03}}{z_{03}},z_{03})=(y,\frac1z)$ & $(1,0)$
  \\  
  $a_5$ & $(y_{31},z_{31})=(\frac{y_{12}}{z_{12}},z_{12})=(\frac1{yz},z)$ & $(\frac2{\theta_0+\eta+\theta_{\infty}},0)$
  \\
   $a_6$ & $(y_{31},z_{31})=(\frac{y_{12}}{z_{12}},z_{12})=(\frac1{yz},z)$ & $(\frac2{\theta_0+\eta-\theta_{\infty}},0)$
  \\  
   $a_7$ & $(y_{41},z_{41})=(\frac{y_{21}}{z_{21}},z_{21})=({yz},\frac1z)$ & $(\theta_0,0)$
  \\   
   $a_8$ & $(y_{52},z_{52})=(y_{21}-1,\frac{z_{21}}{y_{21}-1})=(y-1,\frac1{(y-1)z})$ & $(0,0)$
  \\   
   $a_9$ & $(y_{91},z_{91})=(\frac{y_{52}}{z_{52}},z_{52})=((y-1)^2z,\frac1{(y-1)z})$ & $(\theta_1 e^t,0)$
  \\ 
  $a_{10}$ & $(y_{101},z_{101})=(\frac{y_{91}-\theta_1e^t}{z_{91}},z_{91})$
  & $((1+\eta)\theta_1 e^t,0)$
  \\    
  &$ \phantom{(y_{101},z_{101})} =((y-1)z((y-1)^2z-\theta_1e^t),\frac1{(y-1)z})$&\\
  \hline
\end{longtable}
The constants used throughout the paper are related as follows.
\begin{align*}
&\theta_{\infty}^2=2\alpha,
\\
&\theta_0^2=-2\beta,
\\
&\theta_1^2=-2\delta\qquad(\theta_1\neq0) ,
\\
&\eta=-\frac{\gamma}{\theta_1}-1,
\\
&\epsilon=\frac12(\theta_0+\theta_{\infty}+\eta),
\\
&F=\dfrac12\epsilon(\theta_0+\eta-\theta_{\infty}).
\end{align*}


\end{document}